\documentclass[twoside,leqno,twocolumn]{article}
\newcommand{\func}[1]{\textsc{#1}} %

\usepackage{hyperref}
\usepackage{colortbl}
\usepackage{subcaption}
\usepackage{ltexpprt}
\usepackage[utf8]{inputenc}
\usepackage{amsfonts}

\newcommand{\abs}[1]{\left| #1\right|}

\newcommand{\set}[1]{\left\{ #1\right\}}

\newcommand{\sodass}{\,:\,}
\newcommand{\setGilt}[2]{\left\{ #1\sodass #2\right\}}

\newcommand{\natnull}{\mathbb{N}_{0}}

\newcommand{\realrange}[2]{\left[#1, #2\right]}

\newcommand{\unitrange}[2]{\realrange{0}{1}}

\newcommand{\llabel}[1]{\label{\labelprefix:#1}}
\newcommand{\labelprefix}{} %

\newcommand{\discussionsize}{\small}

\newenvironment{code}{\noindent%
\begin{tabbing}%
\hspace{2em}\=\hspace{2em}\=\hspace{2em}\=\hspace{2em}\=\hspace{2em}\=%
\hspace{2em}\=\hspace{2em}\=\hspace{2em}\=\hspace{2em}\=\hspace{2em}\=%
\kill}{\end{tabbing}}

\newcommand{\labelcommand}{}
\newsavebox{\codeparam}
\newcounter{lineNumber}
\newenvironment{disscodepos}[3]{%
\renewcommand{\labelcommand}{#2}%
\renewcommand{\captiontext}{#3}%
\sbox{\codeparam}{\parbox{\textwidth}{#3}}%
\begin{figure}[#1]\begin{center}\begin{code}\setcounter{lineNumber}{1}}{%
\end{code}\end{center}\caption{\llabel{\labelcommand}\captiontext}\end{figure}}

{\end{disscodepos}}

\newdimen\endofsize\endofsize=0.5em
\def\endofbeweis{~\quad\hglue\hsize minus\hsize
                 \hbox{\vrule height \endofsize width
\endofsize}\par}

\usepackage{epsfig}
\usepackage{url}
\usepackage{amsmath}
\usepackage{amssymb}
\usepackage{microtype}
\usepackage{algorithm}
\usepackage[noend]{algpseudocode}
\usepackage{multicol}
\usepackage{pdflscape}
\usepackage{numprint}

\nprounddigits{2}
\npdecimalsign{.} %
\usepackage{makeidx}
\usepackage[usenames,dvipsnames,table,xcdraw]{xcolor}
\definecolor{lightergray}{rgb}{0.86, 0.86, 0.86}
\definecolor{red}{rgb}{0.698, 0.133, 0.133}

\usepackage{tabularray}
\usepackage{siunitx}
        \UseTblrLibrary{siunitx}
\usepackage{hyperref}
\usepackage{xspace}
\usepackage{wrapfig}
\usepackage{graphics}
\usepackage{todonotes}
\usepackage{booktabs}
\usepackage{microtype}%
\definecolor{infocolor}{rgb}{0.6,0.6,0.6}

\usepackage[numbers,square,sort&compress]{natbib}

\newcommand{\naive}{\textsc{NaiveDynOpt}}
\newcommand{\strong}{\textsc{StrongDynOpt}}
\newcommand{\imp}{\textsc{ImprovedDynOpt}}

\newcommand{\ori}[1]{\overline{#1}}
\DeclareMathOperator{\arboricity}{\alpha}

\newcommand{\ie}{i.\,e.,\xspace}

\newcommand{\etal}{et~al.~}

\usepackage{color}

\def\comment#1{}

\def\withcomments{
  \newcounter{mycommentcounter}
   \def\comment##1{\refstepcounter{mycommentcounter}%
    \ifhmode%
     \unskip%
     {\dimen1=\baselineskip \divide\dimen1 by 2 %
       \raise\dimen1\llap{\tiny\bfseries \textcolor{red}{-\themycommentcounter-}}}\fi%
     \marginpar[{\renewcommand{\baselinestretch}{0.8}%
       \hspace*{3em}\begin{minipage}{5em}\footnotesize [\themycommentcounter]: \raggedright ##1\end{minipage}}]{\renewcommand{\baselinestretch}{0.8}%
       \begin{minipage}{5em}\footnotesize [\themycommentcounter]: \raggedright ##1\end{minipage}}}
  }

\definecolor{darkgreen}{RGB}{0,200,100}
\definecolor{orange}{RGB}{255,80,0}
\definecolor{lightseagreen}{rgb}{0.13, 0.7, 0.67}

\newcommand{\Xcomment}[1]{}

\newcommand{\odeg}[2]{\func{odeg}(#1, #2)}
\newcommand{\preG}{\overline{G}_i}
\newcommand{\postG}{\overline{G}_{i+1}}
\newcommand{\midG}{\tilde{G}_{i}}

\newtheorem{invariant}{Invariant}

\newcommand{\mytitle}{Engineering Fully Dynamic Exact $\Delta$-Orientation Algorithms}

\begin{document}

\title{\Large \mytitle}
\author{Ernestine Großmann, Henrik Reinstädtler, Christian Schulz, Fabian Walliser} 
\date{}

\maketitle

\fancyfoot[R]{\scriptsize{Copyright \textcopyright\ 2024 by E. Großmann, H. Reinstädtler, C. Schulz and F. Walliser}}

\date{}

\maketitle
\begin{abstract}
A (fully) dynamic graph algorithm is a data structure that supports edge insertions, edge deletions, and answers specific queries pertinent to the problem at hand. In this work, we address the fully dynamic edge orientation problem, also known as the fully dynamic $\Delta$-orientation problem. The objective is to maintain an orientation of the edges in an undirected graph such that the out-degree of any vertex remains low. When edges are inserted or deleted, it may be necessary to reorient some edges to prevent vertices from having excessively high out-degrees. In this paper, we introduce the first algorithm that maintains an optimal edge orientation during both insertions and deletions.
In experiments comparing with recent nearly exact algorithms, we achieve a 32\% lower running time.
The update time of our algorithm is up to 6 orders of magnitude faster than static exact algorithms.
\end{abstract}

\section{Introduction}
\label{sec:introduction}

Complex graphs have a wide range of applications, from technological networks to biological systems like the human brain. These graphs can contain billions of vertices and edges, and analyzing them often provides valuable insights. In practice, the structure of these graphs frequently changes over time, with vertices and edges being added or removed. For instance, in a social network, users join or leave, and their connections may form or dissolve, while in road networks, new roads are constructed. A problem is termed \emph{fully dynamic} if it involves both the insertion and deletion of edges.

A (fully) dynamic graph algorithm is a data structure designed to handle edge insertions, edge deletions, and specific problem-related queries. Key dynamic graph problems, such as connectivity, reachability, shortest paths, and matching, are widely studied (see~\cite{DBLP:journals/corr/abs-2102-11169}). Despite the extensive theoretical research on efficient dynamic graph algorithms, many of these have not been implemented or empirically tested. Some classical dynamic algorithms have undergone experimental studies, including early work on (all pairs) shortest paths \cite{DBLP:journals/jea/FrigioniINP98,DBLP:journals/talg/DemetrescuI06}, reachability~\cite{DBLP:conf/alenex/HanauerH020}, and transitive closure~\cite{DBLP:journals/jea/KrommidasZ08,DBLP:conf/wea/HanauerH020}. More recent contributions have focused on fully dynamic clustering (graph, $k$-center)\cite{DBLP:conf/wads/DollHW11,DBLP:conf/alenex/GoranciHLSS21}, fully dynamic approximation of betweenness centrality\cite{DBLP:conf/esa/BergaminiM15}, and fully dynamic minimum cuts~\cite{DBLP:conf/alenex/HenzingerN022}. However, the engineering and practical implementation of these algorithms are still underdeveloped. Many fundamental dynamic graph problems have received little attention from a practical engineering perspective, with few efficient implementations available.

A crucial \emph{building block} for fully dynamic algorithms is storing sparse graphs with \emph{low} memory requirements while still enabling fast adjacency queries. Specifically, for two vertices $u$ and $v$, a function should return true if $\{u,v\}\in E$ and false otherwise, ideally in constant time. Traditional methods for storing dynamic graphs include adjacency matrices, which require $O(n^2)$ space and can answer such queries in $O(1)$ time, or adjacency lists, which require $O(n+m)$ space but may need to search the entire neighborhood of a vertex, potentially leading to slower query times.

For static graphs, Kannan et al.~\cite{DBLP:journals/siamdm/KannanNR92} introduce a method to store an undirected graph efficiently, supporting adjacency queries in $O(\alpha)$ time, where $\alpha$ is the arboricity of the graph.
The \emph{arboricity} $\alpha(G)$ of a graph is the minimum number $t$ such that the graph $G$ can be decomposed into $t$ forests.
The fundamental concept of their algorithm is straightforward: store each edge in the adjacency list of only one of its endpoints. Thus, queries can be executed by checking the adjacency lists of vertex $u$ for vertex $v$, and vice versa.
This approach allows both nodes to maintain relatively short adjacency lists, even for vertices with high degrees.
Storing an edge at a single endpoint equates to giving the edge a direction, making it originate from the stored endpoint.
Formally, an \emph{orientation} of a graph $G=(V,E)$ is defined as a \emph{directed graph} $\ori{G}=(V,E')$ where for every edge $\set{u,v} \in E$, either the directed edge $(u,v)$ or $(v,u)$ is included in $E'$.
To ensure efficient (constant-time) adjacency queries, the out-degree of any vertex in $\ori{G}$ should not exceed a certain constant $\Delta$, leading to the concept of $\Delta$-orientations.

In the static case, an $\alpha$-orientation always exists~\cite{DBLP:journals/gc/ChenMWZZ94, 10.1112/jlms/s1-39.1.12, 10.1112/jlms/s1-36.1.445}.
In practice, one is thus interested in algorithms computing a $\Delta$-orientation with~$\Delta$~being small. The problem of minimizing $\Delta$ can also be solved in polynomial time in static graphs~\cite{venkateswaran2004minimizing}. These algorithms work by finding and reversing paths between vertices with the highest out-degree and those with lower out-degrees iteratively and exhaustively. Additional engineering of these algorithms yields significant performance boosts in practice \cite{reinstaedtler2024engineering}. 

 In this study, we primarily explore the fully dynamic scenario, where maintaining $\Delta$-orientations is crucial for supporting rapid adjacency queries. This capability serves as a fundamental component in numerous dynamic graph algorithms (refer to the detailed list in the extended version of \cite{DBLP:conf/icalp/KopelowitzKPS14}). For instance, Neiman and Solomon \cite{DBLP:journals/talg/NeimanS16} demonstrated the maintenance of a maximal matching within $O(\frac{\log n}{\log \log n})$ amortized time by employing a dynamic edge orientation algorithm. Such algorithms are also instrumental in dynamic matrix vector multiplication \cite{DBLP:conf/icalp/KopelowitzKPS14}, and Kowalik and Kurowski \cite{DBLP:journals/talg/KowalikK06} utilize them to efficiently respond to shortest-path queries up to a specified length $k$ in planar graphs. Additional applications of fully dynamic edge orientations include dynamic coloring \cite{DBLP:conf/icalp/ChristiansenR22} and the tracking of subgraph counts \cite{DBLP:conf/sand/HanauerHH22}. However, only recently have heuristic algorithms for the dynamic problem been \hbox{evaluated in practice~\cite{DBLP:conf/acda/BorowitzG023}.}

\textbf{Contributions:} In this study, we present three novel algorithms designed to maintain an optimal edge orientation amidst edge insertions and deletions. Alongside a naive algorithm, we introduce two advanced algorithms based on novel invariants. Our comprehensive evaluation, conducted on both real-world dynamic graphs and dynamic graphs derived from real-world static graphs, demonstrates significant improvements.
The worst-case complexity of our algorithm for insertion is $\mathcal{O}(m)$, while for deletion it is amortized $\mathcal{O}(m)$. 
Compared to recent nearly-exact algorithms, our approaches achieve a 32\% reduction in running time while ensuring optimal problem-solving in contrast to previous algorithm.
 When comparing with recent static exact algorithms, we achieve an update time 6 orders of magnitude faster than the static running time.

\section{Preliminaries}
\label{sec:preliminaries}

\subsection{Basic Concepts.}
\label{subsec:basic_concepts}

Let $G=(V=\{0,\ldots, n-1\},E)$ be an \emph{undirected graph} of $n$ vertices and $m$ edges. Let $\Gamma(v) = \setGilt{u}{\set{v,u}\in E}$ denote the neighbors of a vertex $v$ and $\deg(v)=\abs{\Gamma(v)}$ the \emph{degree of $v$}.
Further, let $\Delta(G)$ denote the maximum out-degree of $G$.
For a subset of vertices $S$ we define $E(S)$ as the subset of edges having both vertices in $S$.
The density of a (sub-)graph is given by $\lvert E\rvert/\lvert V\rvert$.
A graph-sequence $\mathcal{G} = (G_0, \ldots, G_t)$ for some $t\in\natnull$ is an \emph{edit-sequence of graphs} if there exists
for all $i>0$ an edge $e\in \binom{V}{2}$ such that it is either inserted, \ie $G_{i} = G_{i-1} + e$, or deleted, \ie $G_{i} = G_{i-1} - e$, in update~$i$. 
The \emph{arboricity} $\arboricity(G)$ of a graph is defined as the smallest number $t$ such that $G$ can be partitioned into $t$ forests.
A graph-sequence $\mathcal{G}$ has \emph{bounded arboricity} $\alpha > 0$ if $\arboricity(G)\leq \alpha$ for all $G\in\mathcal{G}$.
A pseudoforest is a forest, where in each connected component there can be one circle.
Similarly, the \emph{pseudoarboricity} is defined as the smallest number of $t$ such that $G$ can be partitioned into $t$ pseudoforests.
An \emph{orientation} of a graph $G=(V,E)$ is a \emph{directed graph} $\ori{G}=(V,E')$ such that for every $\set{u,v}\in E$ either $(u,v)$ or $(v,u)$ is in $E'$. 
The out-degree of a node~$v$ in a directed graph is defined as the number of edges starting in $v$, $\odeg{v}{\ori{G}} = \abs{\setGilt{u}{(v,u)\in E'}}$. By $\Delta$ we refer to the maximum out-degree in $\ori{G}$.
If $\Delta\leq c$, then $\ori{G}$ is a \emph{$c$-orientation} of $G$.
The graph-sequence $\ori{\mathcal{G}} = (\ori{G}_0, \ldots , \ori{G}_t)$ is a \emph{sequence of orientations} of $\mathcal{G}$ if every $\ori{G}_i$ is an orientation of $G_i$. 
In the same way, $\ori{\mathcal{G}}$ is a \emph{sequence of $c$-orientations} if all $\ori{G}_i$ are $c$-orientations.
The goal of the fully dynamic edge orientation problem is to keep the maximum out-degree $\Delta$ minimum at each point in time. 

Given some orientation $\ori{G}$, we say, we \emph{flip} an edge $(u,v)\in E'$ if we delete it and insert $(v,u)$.
Let $P=\langle v_0,\dots, v_k \rangle$ be a \emph{directed path} of length $k$ in $\ori{G}$ \ie there exists an edge $(v_i,v_{i+1})$ in $E'$ for $0 \leq i < k$. $P$ is said to be a \emph{$u$-$v$-path} if it starts in $u$ and ends in $v$.
We \emph{flip} a path by flipping every edge once and denote the obtained result as the \emph{inverse path} $P_f$ of $P$. Two paths, $P_1$ and $P_2$, are said to \emph{share edges} if there is an edge $e \in P_1$ such that $e \in P_2$. $P = \langle u,\dots,v \rangle$ is called an \emph{improving path} if $\odeg{u}{\ori{G}} > \odeg{v}{\ori{G}} + 1$.
We call the vertices with the largest out-degree in an orientation \emph{peak vertices} and vertices $v$ with $\odeg{v}{\ori{G}} < \Delta-1$~\emph{sink~vertices}.

\subsection{Related Work.}
\label{subsec:related_work}
There is a wide range on fully dynamic algorithms in literature in general. The most studied dynamic problems are graph problems such as connectivity, reachability, shortest paths, or matching.
 We refer the reader to the recent survey \cite{DBLP:journals/corr/abs-2102-11169} for more~details.
 In order to understand the underlying problem we give a brief summary of related work on edge orientation and pseudoarboricity.

\subparagraph{Edge Orientation \& Pseudoarboricity.}
Finding an edge orientation is closely related to identifying the densest subgraph and pseudoarboricity. The rounded up density of the densest subgraph equals the pseudoarboricity and thus the optimal out-degree in an edge orientation~\cite{kowalik2006approximation,aichholzer1995optimal,venkateswaran2004minimizing}.
The problem can be $2$--approximated by repeatedly deleting minimum degree vertices in linear time~\cite{charikar2000greedy}.
There are two major techniques to solve the static edge orientation problem exactly.
On the one hand, Venkateswaran~\cite{venkateswaran2004minimizing} provide a framework to solve the problem by finding paths between peak and sink vertices.
On the other hand, Kowalik~\cite{kowalik2006approximation} and Asahiro~et~al.~\cite{asahiro2007graph} present solutions derived from building flow-networks.

Blumenstock~\cite{doi:10.1137/1.9781611974317.10} proved a general worst case bound of $\mathcal{O}(m^{3/2}\sqrt{\log\log \Delta})$ and gave the first experimental evaluation of flow-based approaches.
More recently, Reinstädtler~et~al.~\cite{reinstaedtler2024engineering} presented an alternative flow-based formulation and conducted an extensive experimental study comparing path-based and flow-based algorithms for the problem.
They report significant performance increases by refining path-based techniques.

\subparagraph{Fully Dynamic Edge Orientation Algorithms.}
We now give a high-level overview of dynamic results in the literature. 
To the best of our knowledge, there is no algorithm that maintains an optimum edge orientation under edge updates. All currently available dynamic algorithms are either heuristic or approximate.
Brodal and Fagerberg \cite{DBLP:conf/wads/BrodalF99} were the first to consider the problem in the dynamic case. The authors present a linear space data structure for maintaining graphs with bounded arboricity. The data structure requires a bound $c$ on the arboricity of the graph as input. It then supports adjacency queries in $O(c)$ time, edge insertions in amortized time $O(1)$ as well as edge deletions in amortized time $O(c+\log n)$. The authors note that if the arboricity of a dynamic graph remains bounded, then the forest partitions may change due to the update. To deal with this the authors introduce a re-orientation operation, also called flipping above, which can change the orientation of an edge in order to maintain a small out-degree.
Kowalik \cite{DBLP:journals/ipl/Kowalik07} also needs a bound $c$ on the arboricity. In particular, Kowalik shows that the algorithm of Brodal and Fagerberg can maintain an $O(\alpha \log n)$ orientation of an initially empty graph with arboricity bounded by $c$ in $O(1)$ amortized time for insertions and $O(1)$ worst-case time for deletions. 
Kopelowitz \etal \cite{DBLP:conf/icalp/KopelowitzKPS14} gave an algorithm that does not need a bound on the arboricity as input. Their algorithm maintains an $O(\log n)$-orientation in worst-case update time $O(\log n)$ for any constant arboricity~$\alpha$.
He~\etal\cite{DBLP:conf/isaac/HeTZ14} show how to maintain an $O(\beta \alpha)$-orientation in $O(\frac{\log(n/(\beta \alpha))}{\beta})$ amortized insertion time and $O(\beta \alpha)$ amortized edge deletion time thereby presenting a trade-off between quality of the orientation (the maximum out-degree) and the running time of the operations. Berglin and Brodal \cite{DBLP:journals/algorithmica/BerglinB20} gave an algorithm that allows a \emph{worst-case} user-specific trade-off between out-degree and running time of the operations. Specifically, depending on the user-specified parameters, the algorithm can maintain $O(\alpha+\log n)$ orientation in $O(\log n)$ worst-case time or an $O(\alpha \log^2 n)$-orientation in constant worst-case time.
Banerjee~\etal\cite{BANERJEE20201} present a fully dynamic algorithm for keeping track of the current arboricity of a fully dynamic graph.
Christiansen~\etal\cite{DBLP:journals/corr/abs-2209-14087} published a report which contains algorithms that  make choices based entirely on local information, which makes them automatically adaptive to the current arboricity of the~graph. One of their algorithm maintains a $O(\alpha)$-orientation with worst-case update time $O(\log^2n \log \alpha)$. The authors also provide an algorithm with worst-case update time $O(\log n \log \alpha)$ to maintain an $O(\alpha + \log n)$-orientation.
Recently, Borowitz~\etal\cite{DBLP:conf/acda/BorowitzG023} performed an experimental evaluation for a range of heuristic and approximation algorithms for the fully dynamic problem. Their most competitive with respect to solution quality algorithm is a breadth-first search with limited depth, which they call \textsc{BFS20} and we compare our algorithms with. 

\section{From Optimal Static to Optimal Dynamic Delta-Orientation Algorithms}
\label{s:main}
A variety of methods have been developed to address the edge orientation problem in static graphs. One notable algorithm, proposed by Venkateswaran~\cite{venkateswaran2004minimizing} (see Algorithm~\ref{alg:venkateswaran}), serves as the foundation for the optimal dynamic algorithms introduced in this study. This algorithm begins by arbitrarily assigning an orientation and then iteratively searches for improvements. It identifies paths between a set of vertices, $S$, with a maximum out-degree of $k$, and another set, $T$, with vertices of out-degrees less than $k-1$. 
 If $S$ is empty, the value of $k$ is decremented by one, and the sets $S$ and $T$ are reinitialized. The search terminates when no paths are found, at this point the current $k$ is deemed optimal. The theoretical underpinning of this algorithm is confirmed by examining the density of the subgraph induced by the vertices encountered during an unsuccessful breadth-first search (BFS) originating from $S$. These vertices have at least an out-degree of $k-1$, and at least one has an out-degree of $k$, yielding an average density exceeding $k-1$. This substantiates a pseudoarboricity of $k$ based on subgraph density arguments. Furthermore, the algorithm’s time complexity is established as $\mathcal{O}(m^2)$. Each path is located in $\mathcal{O}(m)$ time, and the potential number of improvements is limited to $m$, constrained by the total number of edges.

In the following, we present three fully dynamic algorithms. First, we introduce a naive extension of Venkateswaran's algorithm to the fully dynamic case, referred to as \naive. This algorithm flips all improving paths after each update. Next, we describe a more efficient algorithm named \strong, which searches for only one improving path per update by maintaining a strong invariant. Finally, we present \imp, an advanced version that relies on a less stringent invariant. This relaxation further reduces the number of improving path \hbox{searches, particularly for insertions.}

\begin{algorithm}[t]
  \begin{algorithmic}[1]
    \Procedure{Venkateswaran}{$G=(V,E)$}
    \State $\ori{G} \gets \textrm{an arbitrary orientation of } G$ 
    \State $k\gets\max_{v\in V}\odeg{v}{\ori{G}}$
    \State $S\gets\{ v\in V \mid \odeg{v}{\ori{G}}= k\}$
    \State $T \gets\{ v\in V \mid \odeg{v}{\ori{G}}\leq k-2\}$
    \While{$\exists$ path $P$=$\langle s, \ldots, t \rangle$ from $S$ to $T$ in $\ori{G}$}
    \State Flip $P$ in $\ori{G}$ 
    \State Remove $s$ from $S$
    \State Remove $t$ from $T$ if $\odeg{t}{\ori{G}} = k-1$
    \If{$S$ empty}
    \State $k\gets k-1$
    \State $S\gets\{ v\in V \mid \odeg{v}{\ori{G}}= k\}$
    \State $T \gets\{ v\in V \mid \odeg{v}{\ori{G}}\leq k-2\}$
    \EndIf
    \EndWhile
    \State\Return $k$
    \EndProcedure
  \end{algorithmic}
  \caption{Static Algorithm by Venkateswaran~\cite{venkateswaran2004minimizing}}
  \label{alg:venkateswaran}
\end{algorithm}
\subsection{The Algorithm \naive.}
The core idea of the naive algorithm is to run Venkateswaran algorithm with the current orientation as a starting point after each update. 
The current objective function value (max out-degree $\Delta$) is maintained in a global variable and vertices are stored in a bucket priority queue by their out-degree.
When an edge~$(u,v)$ is inserted, the algorithm assigns the edge to the adjacency list of $u$.
Furthermore, the algorithm starts a breadth-first search initialized with all maximum out-degree vertices in the directed graph induced by the current orientation and tries to find a vertex $y$ with out-degree strictly smaller than~$\Delta-1$. The breadth-first search algorithm stops as soon as it has found such a vertex.
Assume the algorithm finds such a vertex and let the corresponding path be $p = \langle\phi, \ldots, y\rangle$ where $\phi$ is a maximum out-degree vertex found by following parent pointers, computed by the breadth-first search, starting at $y$. Note that all edges on this path are oriented from $\phi$ to $y$. The algorithm flips each edge on the path, thereby increasing the out-degree of $y$ by one and decreasing the out-degree of $\phi$ by one.
The flipping operations can be done in $\mathcal{O}(|P|)$, as one can store the positions of the target vertices of the edges in the respective adjacency arrays while doing the breadth-first search. 
The algorithm continues searching for improving paths until a search is unsuccessful. Correctness of the algorithm directly follows from the correctness of Venkateswaran's algorithm. 
We arrive at a time complexity of $\mathcal{O}(m^2)$ for a single update operation which is the worst-case running time of Venkateswaran's algorithm. 
Note however that using similar arguments as we use in the following Section~\ref{sec:strong} it is possible to show that inserting an edge can create at most one new improving path from a peak vertex to a sink vertex. Thus, the insert operation of the naive algorithm runs in in $\mathcal{O}(m)$ time.

\newpage
\subsection{The Algorithm \strong.}
\label{sec:strong}
The \naive~algorithm maintains the invariant that there is no improving path between a peak vertex and a sink vertex.
We now present an algorithm that maintains an even stronger invariant. More precisely, the general idea behind \strong, which is presented in Algorithm~\ref{alg:strong}, is maintaining the following invariant.
\begin{invariant}\label{strong_invariant}
There is no improving path between \emph{any} two vertices.
\end{invariant}
\begin{algorithm}[t]
  \caption{Depth-First Find Path Routines}
  \label{alg:find:path}
  \begin{algorithmic}[1]
      \Procedure{FindAndFlipPath}{$u, \ori{G}$}
      \If{${visited}[u]$}
        \State \Return {\textnormal{false}}
      \EndIf
      \ForAll{$v \in \textsc{Adj}[u]$}
         \If{$\odeg{v}{\ori{G}} < \odeg{u}{\ori{G}} - 1$}
          \State flip $(u,v)$
          \State \Return {\textnormal{true}}
        \EndIf
      \EndFor
      \State ${visited}[u] \gets $ true
      \ForAll{$v \in \textsc{Adj}[u]$}
        \If{$\odeg{v}{\ori{G}} = \odeg{u}{\ori{G}} - 1$}
          \If{\Call{FindAndFlipPath}{$v,\ori{G}$}}
            \State flip $(u,v)$ 
            \State \Return {\textnormal{true}}
          \EndIf
        \EndIf
      \EndFor
      \State \Return {\textnormal{false}}
    \EndProcedure
  
    \Procedure{FindAndFlipPathRev}{$u, \ori{G}$}
      \If{${visited}[u]$}
        \State \Return {\textnormal{false}}
      \EndIf
      \ForAll{$v$ with $u \in \textsc{Adj}[v]$}
        \If{$\odeg{v}{\ori{G}} > \odeg{u}{\ori{G}} + 1$}
          \State flip $(v,u)$
          \State \Return {\textnormal{true}}
        \EndIf
      \EndFor
      \State \emph{visited}$[u] \gets $ true
      \ForAll{$v$ with $u \in \textsc{Adj}[v]$}
        \If{$\odeg{v}{\ori{G}} = \odeg{u}{\ori{G}} + 1$}
          \If{\Call{FindAndFlipPathRev}{$v,\ori{G}$}}
            \State flip $(v,u)$
            \State \Return {\textnormal{true}}
          \EndIf
        \EndIf
      \EndFor
      \State \Return {\textnormal{false}}
    \EndProcedure
  \end{algorithmic}
  \end{algorithm}
As we will show, it is possible to maintain Invariant~\ref{strong_invariant} using only one improving path search from a single node per update operation. We achieve this by employing the two functions described in Algorithm~\ref{alg:find:path}. \func{FindAndFlipPath}$(u,\ori{G})$ represents a simple depth-first search that can find and flip an improving path in $\ori{G}$ that starts from a node $u$. Similarly, \func{FindAndFlipPathRev}$(u,\ori{G})$ is able to do the same with $u$ as the end node of the found path by conducting a search on the reverse orientation of~$\ori{G}$. Further enhancements are discussed in Section~\ref{sec:path_finding}.

Upon insertion, the newly added edge is initially oriented from the node with the lower out-degree. Subsequently, an improving path is sought from that node by calling \func{FindAndFlipPath}. Conversely, during deletion, the edge is removed followed by a reverse improving path search using \func{FindAndFlipPathRev} from the node the edge was originally oriented from. This search aims to find another node with a higher out-degree that can now be improved by flipping this path.

In the following, we first demonstrate that Invariant~\ref{strong_invariant} results in an optimal edge orientation and then proceed to prove that Algorithm~\ref{alg:strong}~maintains Invariant~\ref{strong_invariant}. Through this proof, we establish that it is only necessary to flip \emph{at most one} improving path after each modification to the graph to guarantee an optimal solution.

\begin{algorithm}[t]
	\caption{\strong}
	\label{alg:strong}
	\begin{algorithmic}[1]
		\Procedure{Insert}{$u, v, \ori{G}$}
		\State W.l.o.g.: $\odeg{u}{\ori{G}} \leq \odeg{v}{\ori{G}}$
		\State $\textsc{Adj}[$u$] := \textsc{Adj}[u] \cup\{v\}$
		\State $\func{FindAndFlipPath}(u, \ori{G})$
		\EndProcedure
		\Procedure{Delete}{$u, v, \ori{G}$}
		\State W.l.o.g. edge is directed from $u$ to $v$
		\State $\textsc{Adj}[u] := \textsc{Adj}[u] \setminus \{v\}$
		\State $\func{FindAndFlipPathRev}(u, \ori{G})$
		\EndProcedure
	\end{algorithmic}
\end{algorithm}

\vspace{-0.25cm}
\subsubsection{Proof of Correctness.}\label{sec:strong_proof_of_correctness}
We now show that the optimality of the orientation 
is guaranteed when maintaining Invariant~\ref{strong_invariant}. To show this, we first introduce the subsequent auxiliary lemma.

\begin{lemma}[Venkateswaran~\cite{venkateswaran2004minimizing}]\label{lem:help_venkateswaran}
Given a graph $G=(V,E)$, if an edge orientation $\ori{G}$ has  maximum \hbox{out-degree $\Delta$} and there is a subset $S$ of vertices such that $\Delta = \lceil \lvert E(S)\rvert / \lvert S\rvert \rceil$, then $\Delta$ is the optimum out-degree.
\end{lemma}
With the help of Lemma~\ref{lem:help_venkateswaran}, we can now prove the next statement by following Venkateswaran~\cite{venkateswaran2004minimizing} for the correctness of his extremal orientation algorithm.

\begin{theorem}\label{thm:optimality_invariant}
For a given graph $G=(V, E)$ and an edge orientation $\ori{G}$ satisfying  Invariant~\ref{strong_invariant} the resulting maximum out-degree $\Delta$ is optimal.
\end{theorem}

\begin{proof}[Theorem~\ref{thm:optimality_invariant}].
Let $S = \{v \in V \mid \func{odeg}(v,\ori{G}) = \Delta\}$ and \hbox{$T = \{v \in V \mid \func{odeg}(v,\ori{G}) \leq \Delta-2\}$}. Let $U$ be the set of $S$ and all nodes reachable from $S$ by direct paths. 
Invariant~\ref{strong_invariant} implies, that for all $u\in U$ it holds $\func{odeg}(u,\ori{G}) \geq \Delta-1$. Thus follows \hbox{$\lvert E(U)\rvert = \sum_{u \in U} \func{odeg}(u,\ori{G}) > \lvert U\rvert(\Delta - 1)$} since at least one node in $U$ has out-degree $\Delta$. 
Therefore, $\Delta = \lceil \lvert E(U)\rvert / \lvert U\rvert \rceil$ and with Lemma~\ref{lem:help_venkateswaran} follows \hbox{that $\Delta$ is optimal.}
\end{proof}

As a next step, we show that Algorithm \ref{alg:strong} maintains Invariant~\ref{strong_invariant} with every update operation by searching for one improving path and flipping it on success. We demonstrate this in Theorem~\ref{thm:invariant_holds}. Before that we present some preparatory lemmas.

In the following, let $\preG$ denote the orientation in our sequence before and $\postG$ after the update operation. Further, let $\midG$ define the orientation resulting from $\preG$ after the graph update, but before edges have been reoriented. Note that  $\midG$ is not \hbox{necessarily an optimal orientation.}

\begin{lemma}\label{lem:find_path} 
Let $\preG$ satisfy Invariant~\ref{strong_invariant} and let $\{u,v\}$ be the updated edge. If Invariant~\ref{strong_invariant} is not satisfied for $\midG$, we find an improving path when running the routines in Algorithm~\ref{alg:strong}. The path starts at $u$ in case of an insertion and ends in $u$ in case of a deletion. 
\end{lemma}

\begin{proof}[Lemma~\ref{lem:find_path}].
W.l.o.g., let $\odeg{u}{\preG} \leq \odeg{v}{\preG}$ for an insertion and $\{u,v\}$ be oriented from $u$ in $\preG$ for a deletion. During the transition from $\preG$ to $\midG$, only the out-degree of $u$ changes and $(u,v)$ being inserted in the case of an insertion. Therefore, Invariant~\ref{strong_invariant} is either satisfied, or there is an improving path $P = \langle u, \dots, w \rangle$ found by \func{FindAndFlipPath}$(u,\midG)$ starting from $u$ in the case of an insertion, or $P = \langle w, \dots, u \rangle$ found by \func{FindAndFlipPathRev}$(u,\midG)$ ending in $u$ for a deletion.
\end{proof}

\newcommand{\improvpath}{P_\mathcal{I}}
\begin{lemma}\label{lem:one_shared_node}
Let $\preG$ satisfy Invariant~\ref{strong_invariant}, $\{u,v\}$ be the updated edge and $P$ be an improving path found by the routines in Algorithm~\ref{alg:strong}. 
Furthermore, let $P_f$ be the \textbf{f}lipped path of $P$ in $\postG$.
Then, every improving path $\improvpath$ in $\postG$ shares at least one node with $P_f$.
\end{lemma}

\begin{proof}[Lemma~\ref{lem:one_shared_node}].
W.l.o.g., let $\odeg{u}{\preG} \leq \odeg{v}{\preG}$ for an insertion and $\{u,v\}$ be oriented from $u$ in $\preG$ for a deletion. Lemma~\ref{lem:find_path} implies, that $P = \langle u, \dots, w \rangle$ in case of an insertion and $P = \langle w, \dots, u \rangle$ in case of a deletion for some $w \in V$.
Updating the orientation from $\preG$ to $\postG$ involves changing $\odeg{w}{\preG}$, flipping the edges in $P$, and adding $(u,v)$ in case of an insertion. By Invariant~\ref{strong_invariant}, an improving path $\improvpath$ in $\postG$ must either include $w$, share edges with $P_f$, or contain $(u,v)$. Since $u, w \in P_f$, $P_f$ and $\improvpath$ always share at least one node.
\end{proof}

The following lemmas demonstrate that once an improving path has been found, there can not be another improving path. We first address the case related to insertions and then proceed with deletions. 

\begin{lemma}\label{lem:ins}
Let $\preG$ satisfy Invariant~\ref{strong_invariant}, $\{u,v\}$ be the \emph{inserted} edge and $P$ be an improving path found by \func{FindAndFlipPath} after insertion in Algorithm~\ref{alg:strong}. Then, there cannot be an improving path $\improvpath$ in $\postG$.
\end{lemma}

\begin{figure}
  \centering
  \includegraphics{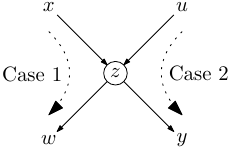}
  \caption{Visualization of the two (most important) cases in the proof for Lemma~\ref{lem:ins}. The path $u \to z \to w$ is the improving path found after inserting the edge, the path $x \to z \to y$ is the assumed improving path in $\postG$. Case 1 depicts the case in which $\odeg{x}{\postG} > \odeg{u}{\postG}$. In this case an improving path $x \to w$ can be found in $\preG$ if there is an improving path in $\postG$. Case 2 shows the case in which $\odeg{x}{\postG} \leq \odeg{u}{\postG}$. In this case an improving path $u \to y$ can be found in $\preG$ if there is an improving path in $\postG$. In both cases, this yields a contradiction to Invariant~\ref{strong_invariant}.}
  \label{fig:2cases}
  \vspace{-0.5cm}
\end{figure}

\begin{proof}[Lemma~\ref{lem:ins}].
W.l.o.g., let $\odeg{u}{\preG} \leq \odeg{v}{\preG}$. From Lemma~\ref{lem:find_path}, we derive that $P = \langle u, \dots, w \rangle$. Further, let $P_f$ be the path \textbf{f}lipped in $\postG$ such that $\odeg{v}{\postG} \geq \odeg{u}{\postG} \geq \odeg{w}{\postG}$. 
Assume there exists an improving path $\improvpath = \langle x, \dots, y \rangle$ in $\postG$. Based on this, we demonstrate that we can find a path in $\preG$ as shown in Figure~\ref{fig:2cases} consisting of sub-paths from $\improvpath$ and $P$, which contradicts Invariant~\ref{strong_invariant}. 
We distinguish the following cases.  

\emph{Case 1:} $\odeg{x}{\postG} > \odeg{u}{\postG}$. Roughly speaking, if this is the case, then we show that there is an improving path from $x$ to $w$ in $\preG$ which is a contradiction to $\preG$ satisfying Invariant 1. 

From Lemma~\ref{lem:one_shared_node} it follows that $P_f$ and $\improvpath$ have at least one common node.
Let $z$ be the first such node in $\improvpath$, resulting in $\improvpath=\langle x,\dots,z,\dots,y \rangle$ and $P=\langle u\dots,z,\dots,w \rangle$. 
Therefore, the sub-paths $P_{x \to z} := \langle x,\dots,z \rangle$ of $\improvpath$ and $P_{z \to w} := \langle z,\dots,w \rangle$ of $P$ both share no edges with the flipped path $P_f$ and $(u,v) \notin P_{x \to z}$. 

In the case of $z \neq u$ or $(u,v) \notin P$, it also follows that $(u,v) \notin P_{z\to w}$.  Therefore, $P_{x \to z}$ and $P_{z\to w}$ already exist in $\preG$. This implies, that there is an improving path $\langle x,\dots,z,\dots,w \rangle$ in $\preG$ with $\odeg{x}{\preG} = \odeg{x}{\postG} > \odeg{u}{\postG} \geq \odeg{w}{\postG} = \odeg{w}{\preG} + 1$, which contradicts Invariant~\ref{strong_invariant}.

If $z = u$ and $(u,v) \in P$, then $(u,v) \notin \improvpath$ because $(u,v)$ got flipped. Therefore, $\improvpath$ has to contain $w$ or share an edge with $P_f$ according to our previous assumption. This implies there is at least one other node apart from~$z(=u)$, which is in $\improvpath$ and $P_f$, as $w \in P_f$. Let $\tilde{z} \neq z$ be the first such node, resulting in $\improvpath=\langle x,\dots,z,\dots,\tilde{z},\dots,y \rangle$ and $P=\langle u,v,\dots,\tilde{z},\dots,w \rangle$. Therefore, the sub-paths $P_{x\to \tilde z} := \langle x,\dots,\tilde{z} \rangle$ of $\improvpath$ and $P_{\tilde z \to w} := \langle \tilde{z},\dots,w \rangle$ of $P$ both share no edges with $P_f$ and $(u,v) \notin P_{x \to \tilde z}$, $(u,v) \notin P_{\tilde z \to w}$. It follows there is a path $\langle x,\dots,\tilde{z},\dots,w \rangle$ in $\preG$ with again $\odeg{x}{\preG} > \odeg{w}{\preG} + 1$, which contradicts Invariant~\ref{strong_invariant}. Therefore, $\odeg{x}{\postG} > \odeg{u}{\postG}$ is not possible.

\emph{Case 2:} $\odeg{x}{\postG} \leq \odeg{u}{\postG}$. Roughly speaking, in this case we can construct an improving path from $u$ to $y$ or one from $v$ to $y$ in $\preG$ which yields a contradiction to $\preG$ satisfying Invariant~\ref{strong_invariant}.
        
According to Lemma~\ref{lem:one_shared_node}, the improving path $\improvpath$ in $\postG$ and the improving path $P$  in $\preG$ share at least one node. Let $z$ be the last such node in $\improvpath$,
resulting in $\improvpath=\langle x,\dots,z,\dots,y \rangle$ and $P=\langle u,\dots,z,\dots,w \rangle$. Therefore, the sub-path $P_{z \to y} := \langle z,\dots,y \rangle$ of the improving path $\improvpath$ does not share edges with $P_f$ (the flipped path in $\postG$).
 Furthermore, $v$ is not part of $P_{z \to y}$, because otherwise there is a sub-path $P_{v \to y} := \langle v,\dots,y \rangle$ of $P_{z \to y}$, which also does not share any edges with $P_f$ and $(u,v) \notin P_{z \to y}$, therefore contradicting Invariant~\ref{strong_invariant} since $\odeg{v}{\preG} = \odeg{v}{\postG} \geq \odeg{u}{\postG} \geq \odeg{x}{\postG} \geq \odeg{y}{\postG} + 2 \geq \odeg{y}{\preG} + 2$.
This implies that $(u,v) \notin P_{z\to y}$ and further that $P_{z\to y}$ exists in $\preG$.

In the case of $(u,v) \in P$, $P=\langle u,v,\dots,z,\dots,w \rangle$ and the sub-path $P_{v \to z } := \langle v,\dots,z \rangle$ of $P$ does not share edges with $P_f$ and $(u,v) \notin P_{v \to z}$.
It follows that there is a path $\langle v,\dots,z,\dots,y \rangle$ in $\preG$ with $\odeg{v}{\preG} \geq \odeg{y}{\preG}+2$ again, which contradicts Invariant~\ref{strong_invariant}.

If $(u,v) \notin P$, then the sub-path $P_{u \to z} := \langle u,\dots,z \rangle$ of $P$ also shares no edges with $P_f$ and $(u,v) \notin P_{u \to z}$.
It follows that there is a path $\langle u,\dots,z,\dots,y \rangle$ in $\preG$ with \hbox{$\odeg{u}{\preG} = \odeg{u}{\postG} \geq \odeg{x}{\postG} \geq$} \hbox{$\odeg{y}{\postG} + 2 \geq \odeg{y}{\preG} + 2$}, which is a contradiction to Invariant~\ref{strong_invariant}.
 
Since $\odeg{x}{\postG} \leq \odeg{u}{\postG}$ is also not possible, we demonstrated that there can not be an improving path in $\postG$ after the update operation has been performed.
\end{proof}

\begin{lemma}\label{lem:del}
Let $\preG$ satisfy Invariant~\ref{strong_invariant}, $\{u,v\}$ be the deleted edge and $P$ be an improving path found by \func{FindAndFlipPathRev} in Algorithm~\ref{alg:strong}. Then, there cannot be an improving path $\improvpath$ in $\postG$.
\end{lemma}

\begin{proof}[Lemma~\ref{lem:del}].
W.l.o.g., let $\{u,v\}$ be oriented from $u$ in $\preG$. From Lemma~\ref{lem:find_path}, we derive that $P = \langle w, \dots, u \rangle$. Further, let $P_f$ be the path flipped in $\postG$ such that $\odeg{w}{\postG} \geq \odeg{u}{\postG}$. Assume there exists an improving path $\improvpath = \langle x, \dots, y \rangle$ in $\postG$. Based on this, we demonstrate that we can find a path in $\preG$ consisting of sub-paths from $\improvpath$ and $P_f$, which contradicts Invariant~\ref{strong_invariant}. As before we distinguish \hbox{the following cases.}

\emph{Case 1:} $\odeg{x}{\postG} > \odeg{w}{\postG} + 1$. From Lemma~\ref{lem:one_shared_node} it follows that $P_f$ and $\improvpath$ have at least one common node.
Let $z$ be the first such vertex in $\improvpath$, resulting in $\improvpath = \langle x,\dots,z,\dots,y \rangle$ and $P = \langle w,\dots,z,\dots,u \rangle$. Therefore, the sub-paths $P_{x \to z} := \langle x,\dots,z \rangle$ of $\improvpath$ and $P_{z \to u} := \langle z,\dots,u \rangle$ of $P$ both share no edges with $P_f$. This implies that there is a path $\langle x,\dots,z,\dots,u \rangle$ in $\preG$ with $\odeg{x}{\preG} = \odeg{x}{\postG} > \odeg{w}{\postG} + 1 \geq  \odeg{u}{\postG} + 1 = \odeg{u}{\preG} + 1$, which is an improving path in $\preG$ and a contradiction to Invariant~\ref{strong_invariant}. Therefore, $\odeg{x}{\postG} > \odeg{w}{\postG} + 1$ is not possible.

\emph{Case 2:} $\odeg{x}{\postG} \leq \odeg{w}{\postG} + 1$. Since $\improvpath$ and $P$ share nodes according to Lemma~\ref{lem:one_shared_node}, let $z$ be the last such node in $\improvpath$,
resulting in $\improvpath = \langle x,\dots,z,\dots,y \rangle$ and $P = \langle w,\dots,z,\dots,u \rangle$. Therefore, the sub-paths $P_{z \to y} = \langle z,\dots,y \rangle$ of $\improvpath$ and $P_{w \to z} := \langle w,\dots,z \rangle$ of $P$ both share no edges with $P_f$. This implies that there is a path $\langle w,\dots,z,\dots,y \rangle$ in $\preG$ with $\odeg{w}{\preG} = \odeg{w}{\postG} + 1 \geq \odeg{x}{\postG} \geq \odeg{y}{\postG} + 2 = \odeg{y}{\preG} + 2$, which is also a contradiction to Invariant~\ref{strong_invariant}. Since $\odeg{x}{\postG} \leq \odeg{w}{\postG} + 1$ is also not possible, we demonstrated that there can not be an improving path in $\postG$.
\end{proof}

We can now show that Algorithm~\ref{alg:strong} maintains Invariant~\ref{strong_invariant}. 

\begin{theorem}\label{thm:invariant_holds}
Let $\preG$ be an orientation satisfying Invariant~\ref{strong_invariant}. After an update operation and applying the update procedures of \strong (see Algorithm~\ref{alg:strong}), 
Invariant~\ref{strong_invariant} is satisfied for $\postG$.
\end{theorem}

\begin{proof}[Theorem~\ref{thm:invariant_holds}].
Using Lemma~\ref{lem:find_path}, we always find and flip one improving path $P$ if the orientation is not already satisfying Invariant~\ref{strong_invariant}. In the case of insertions, Lemma~\ref{lem:ins} now implies that there can be no further improving path in $\postG$. Analogously, in the case of deletions, we derive from Lemma~\ref{lem:del} that there can be no further improving path in $\postG$. Thus, Invariant~\ref{strong_invariant} is satisfied for $\postG$ and Theorem~\ref{thm:optimality_invariant} implies that $\postG$ is an optimal edge orientation.
\end{proof}

This concludes the correctness proof of our algorithm \strong. We have demonstrated that the algorithm maintains Invariant~\ref{strong_invariant} for both edge deletions and insertions, as established in Theorem~\ref{thm:invariant_holds}. Consequently, it computes optimal solutions, as shown in Theorem~\ref{thm:optimality_invariant}.

\subsection{The Algorithm \imp.}
We now propose an improved algorithm, which maintains a less strict version of Invariant~\ref{strong_invariant}. The specific pseudocode for this algorithm is provided in Algorithm~\ref{improved_invariant}, and the new invariant is defined as follows.
\begin{invariant}\label{improved_invariant}
There is no improving path for 
any vertex with maximum out-degree.
\end{invariant}
Here, we maintain the maximum out-degree $\Delta$ and the number of peak vertices $m_c$ as global variables. In the following we explain the different \hbox{update procedures in detail.}

\newcommand{\onePathImproved}{one}
\newcommand{\allPathsImproved}{all}
\begin{algorithm}[bt!]
\caption{Improved Fully Dynamic Algorithm}
\label{alg:dyn_opt_ins}
\label{alg:dyn_opt_solve}
\begin{algorithmic}[1]
    \State global variables: \\{$\Delta$} $= 0$ \Comment{max out-degree} \\{$m_c$} $= 0$ \Comment{\# peak vertices} 
  
  \Procedure{Insert}{$u, v, \ori{G}$}
    \State W.l.o.g.: $\odeg{u}{\ori{G}} \leq \odeg{v}{\ori{G}}$
    \State $\textsc{Adj}[$u$] := \textsc{Adj}[u] \cup\{v\}$
    
    \If{$\odeg{u}{\ori{G}} = \Delta$}
      \If{not $\func{FindAndFlipPath}(u, \ori{G})$}
        \State {$m_c\texttt{++}$}
      \EndIf
    \ElsIf{$\odeg{u}{\ori{G}} = {\Delta} + 1$}
      \If{$\func{FindAndFlipPath}(u, \ori{G})$} \label{alg:dyn_opt_ins:step2}
        \State $m_c\texttt{++}$
      \Else
        \State $\Delta\texttt{++}$, $m_c = 1$
      \EndIf
    \EndIf
  \EndProcedure

  \Procedure{Delete}{$u, v, \ori{G}$}
    \State W.l.o.g. edge is directed from $u$ to $v$
    \State $\textsc{Adj}[u] := \textsc{Adj}[u] \setminus \{v\}$
    
    \If{$\odeg{u}{\ori{G}} = \Delta - 1$}
      \State $m_c\texttt{--}$
    \ElsIf{$\odeg{u}{\ori{G}} = \Delta - 2$}
      \If{$\func{FindAndFlipPathRev}(u, \ori{G})$}
        \State $m_c\texttt{--}$
      \EndIf
    \EndIf
    
    \If{${m_c} = 0$} \label{alg:dyn_opt_del:m_c}
      \State ${\Delta}\texttt{--}$
      \State \textsc{TightenOutdegree}$(\ori{G})$ 
      \label{alg:dyn_opt_del:step5}
    \EndIf
  \EndProcedure
\Procedure{TightenOutdegree}{$\ori{G}$}
\State ${\onePathImproved} \gets $ True \Comment{one path improved}
\While{\onePathImproved}
    \State ${\allPathsImproved} \gets $ True \Comment{all paths improved}
    \State ${\onePathImproved} \gets $ False; $m_c = 0$
    \For{$v\in V$ with $\odeg{v}{\ori{G}}=\Delta$}
        \If{\textsc{FindAndFlipPath}($v,\ori{G}$)}
            \State ${\onePathImproved} \gets $ True
        \Else
            \State ${\allPathsImproved} \gets $ False; $m_c\texttt{++}$
        \EndIf
    \EndFor
    \If{\allPathsImproved}
        \State $\Delta\texttt{--}$
    \EndIf
\EndWhile
\EndProcedure
\end{algorithmic}
\end{algorithm}

\subparagraph{Insertion.}
First, the algorithm inserts an edge $e=\{u,v\}$ preliminarily with an orientation $(u,v)$ such that $\odeg{u}{\preG} \leq \odeg{v}{\preG}$. If this node subsequently has an out-degree less than {$\Delta$}, the function terminates immediately, as Invariant~\ref{improved_invariant} already holds, since~$u$ has been a sink vertex before insertion and any improving path using the newly inserted edge would also yield an improving path stopping at $u$.

If the out-degree of $u$ is equal to $\Delta$, it attempts to find and flip an augmenting path starting from $u$ as this is now a new node with maximum out-degree. If this attempt fails, it increments the counter $m_c$. 

If $\odeg{u}{\postG}=\Delta+1$, the algorithm also attempts to find and flip an augmenting path. If successful, it increments $m_c$ as the sink node of the path has to be a new node with maximum out-degree. Otherwise, it increments $\Delta$ by one and resets $m_c$ to one as there is now one node with a new maximum out-degree $\Delta+1$. For all cases, the algorithm requires at most one DFS to find an augmenting path and preserves Invariant~\ref{improved_invariant}.
\subparagraph{Deletion.}
The deletion operation of Algorithm~\ref{alg:dyn_opt_solve} is initiated by removing the edge $e=\{u,v\}$, oriented from $u$ to $v$. If the out-degree of $u$ is subsequently less than ${\Delta} - 2$, then the node has been a sink node before and has not been reachable by a peak node because of Invariant~\ref{improved_invariant}.
Thus, no further search is needed to maintain optimality.

In the case $\odeg{u}{\midG}= {\Delta} - 1$, the out-degree of a peak node has been decreased by one, and  ${m_c}$ needs to be decremented. 
If the out-degree of $u$ after removal is $\Delta - 2$, the deletion of the edge has created a new sink vertex. Thus, we perform a backward search to try to find a path $\langle y,\dots u\rangle$ with $y$ being a peak vertex. If we find such a path, the corresponding edges are flipped. Successful identification and subsequent path flipping from a peak node also lead to a decrement of ${m_c}$. 
If the count of peak nodes decreases to zero, recomputation of a minimal $m_c$ is necessary. 
This is done by exhaustively searching for and flipping improving paths for each new peak vertex in \textsc{TightenOutdegree}. 
More precisely, it repeatedly attempts to find and flip paths for nodes with a current max out-degree $\Delta$. If all paths are improved in an iteration, $\Delta$ is decremented. The process continues until no more paths can be improved.

\subsubsection{Proof of Correctness.}
The implied optimality under the invariant remains unchanged. This can be shown analogously to Theorem~\ref{thm:optimality_invariant} when using Invariant~\ref{improved_invariant}.
Furthermore, we now demonstrate again that this new invariant is preserved when the procedures in Algorithm~\ref{alg:dyn_opt_solve} are called.
Accordingly, the following lemmas show that once an improving path has been found, there can not be another improving path from a peak node. 
We first address the insertions and then proceed with deletions.

\begin{lemma}\label{lem:improved_insertion}
Let $\preG$ satisfy Invariant~\ref{improved_invariant} and $\{u,v\}$ be the inserted edge, then Invariant~\ref{improved_invariant} holds for $\postG$ after applying \func{Insert}$(u,v,\preG)$ from Algorithm~\ref{alg:dyn_opt_solve}.
\end{lemma}

\begin{proof}[Lemma~\ref{lem:improved_insertion}].
W.l.o.g., let $\odeg{u}{\preG} \leq \odeg{v}{\preG}$. 
Since the update depends on $\odeg{u}{\midG}$, we examine each case independently. 
As before, let $\midG$ define the orientation resulting from $\preG$ after the graph update, but before edges have been reoriented.

If $\odeg{u}{\midG} < \Delta$, then Invariant~\ref{improved_invariant} immediately holds for $\postG$ since only the out-degree of $u$ changed, $(u,v)$ got inserted and there can not be a path from a peak node to $u$ in $\preG$ according to Invariant~\ref{improved_invariant}. 

If $\odeg{u}{\midG} = \Delta$ and \func{FindAndFlipPath}$(u,\midG)$ is successful, then $\Delta$ and $m_c$ remain unchanged. It can be shown, analogous to case 1 of the proof for Lemma~\ref{lem:ins}, that Invariant~\ref{improved_invariant} is preserved for $\postG$, except that we apply the lemma only to improving paths starting from a peak node and use Invariant~\ref{improved_invariant} for the proof.

If $\odeg{u}{\midG} = \Delta$ and \func{FindAndFlipPath}$(u,\midG)$ is unsuccessful, then Invariant~\ref{improved_invariant} holds for $\postG$ since only the out-degree of $u$ has changed and $(u,v)$ has been inserted. As $u$ becomes a new peak \hbox{node, $m_c$ is increased.}

If $\odeg{u}{\midG} = \Delta + 1$ and \func{FindAndFlipPath}$(u,\midG)$ successfully finds an improving path $P = \langle u, \dots, w \rangle$, then there is a path from either $u$ or $v$ to $w$ in $\preG$, which are both peak nodes. Therefore, $m_c$ increases, as Invariant~\ref{improved_invariant} implies that $\odeg{w}{\postG} = \Delta$. Similar to case 2 of the proof for Lemma~\ref{lem:ins}, but applied only to improving paths starting from a peak node and using Invariant~\ref{improved_invariant}, it follows that Invariant~\ref{improved_invariant} \hbox{is preserved for $\postG$.}

If $\odeg{u}{\midG} = \Delta + 1$ and \func{FindAndFlipPath}$(u,\midG)$ is unsuccessful, then $u$ is the only peak node and Invariant~\ref{improved_invariant} holds for $\postG$. Consequently, $\Delta$ is increased and $m_c$ set to $1$. Thus Invariant~\ref{improved_invariant} is maintained for all \hbox{occurring cases of insertions.}
\end{proof}

\begin{lemma}\label{lem:improved_deletion}
Let $\preG$ satisfy Invariant~\ref{improved_invariant} and $\{u,v\}$ be the deleted edge, then Invariant~\ref{improved_invariant} holds for $\postG$ after applying \func{Delete}$(u,v,\preG)$ from Algorithm~\ref{alg:dyn_opt_solve}. 
\end{lemma}

\begin{proof}[Lemma~\ref{lem:improved_deletion}].
W.l.o.g., let $\{u,v\}$ be oriented from $u$ in $\preG$. 
Again, $\midG$ defines the orientation resulting from $\preG$ after the graph update, but before edges have been reoriented.
It is immediate from Invariant~\ref{improved_invariant} that $\odeg{u}{\midG} < \Delta$. Further, if the algorithm reduces $m_c$ to $0$, $\Delta$ decreases, prompting an exhaustive search to regain Invariant~\ref{improved_invariant} and recompute $\delta$, $m_c$. 
Accordingly, in the following, we only consider the cases where $m_c > 0$ in line~\ref{alg:dyn_opt_del:m_c}. Since further actions depend on $\odeg{u}{\midG}$, we examine each remaining case independently. 

If $\odeg{u}{\midG} < \Delta + 2$, then Invariant~\ref{improved_invariant} immediately holds for $\postG$ since only the out-degree of $u$ changed and there can not be a path from a peak node to $u$ in $\preG$ according to Invariant~\ref{improved_invariant}.

If $\odeg{u}{\midG} = \Delta - 2$ and \func{FindAndFlipPathRev}$(u,\midG)$ is successful, then the amount of peak nodes $m_c$ is decreased by one. 
Similar to case 2 of the proof for Lemma~\ref{lem:del}, but applied only to improving paths starting from a peak node and using Invariant~\ref{improved_invariant}, it follows that Invariant~\ref{improved_invariant} is preserved for $\postG$ if $m_c > 0$.

If $\odeg{u}{\midG} = \Delta - 2$ and \func{FindAndFlipPathRev}$(u,\midG)$ is unsuccessful, then Invariant~\ref{improved_invariant} holds for $\postG$ as only the out-degree of $u$ has changed.

If $\odeg{u}{\midG} = \Delta - 1$, then $m_c$ is decreased as $u$ is no longer a peak node. Since $u$ can not be the end of an improving path and only the out-degree of $u$ changed, Invariant~\ref{improved_invariant} holds for $\postG$ if $m_c > 0$. We have thus shown that Invariant~\ref{improved_invariant} is maintained for all occurring cases of deletions.
\end{proof}

This concludes the correctness proof of our algorithm \imp. We have demonstrated that the algorithm maintains Invariant~\ref{improved_invariant} for both edge insertions and deletions. Consequently, it maintains optimal solutions, as is shown analogous to Theorem~\ref{thm:optimality_invariant}.

\subsection{(Worst-Case) Complexity of \textsc{StrongDynOpt} and \textsc{ImprovedDynOpt}.}
The complexity of all insertion procedures is bounded by the path search that is required to restore the respective invariant.
 In general, a path search can be done in $\mathcal{O}(m)$. 
 Since the \textsc{StrongDynOpt} performs a path search on every update, its complexity is $\mathcal{O}(m)$ for both, insertion and deletions.

 The analysis for \textsc{ImprovedDynOpt} requires more attention.
 The insertion case is similar, the most expensive operation is a path search, if a peak vertex is generated by an insertion.
 In the case of a deletion, the most expensive operation is the \textsc{TightenOutdegree} procedure, that is called if the maximum out-degree is to be reduced.
 The complexity of  \textsc{TightenOutdegree} can be bounded by $\mathcal{O}(mn)$.
 However, we can amortize its cost by accounting path searches that have been not run/successful in the insertion case.
 Every improving path found by \textsc{TightenOutdegree} can be attributed to a previous insertion of an edge in its path, that did not require a path search on insertion.

Since all update operations cost amortized $\mathcal{O}(m)$, the total cost of all update operations of the algorithm is $\mathcal{O}(m^2)$ and matches the complexity of the simple algorithm by Venkateswaran~\cite{venkateswaran2004minimizing}.

\begin{figure*}[t]

\vspace*{-.75cm}
\caption{Performance Profiles.}
  \begin{subfigure}[t]{0.3\textwidth}
    \centering
    \includegraphics[width=0.95\textwidth]{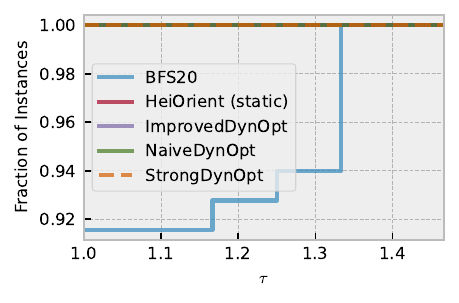}
    \caption{$\Delta$ performance profile comparing our approaches with \textsc{BFS20} by Borowitz~\etal\cite{DBLP:conf/acda/BorowitzG023} and \textsc{HeiOrient}. The \textsc{BFS20} approach is the only inexact solver.}
    \label{fig:results:size}
  \end{subfigure}\hfill
  \begin{subfigure}[t]{0.3\textwidth}
    \centering
  \includegraphics[width=0.95\textwidth]{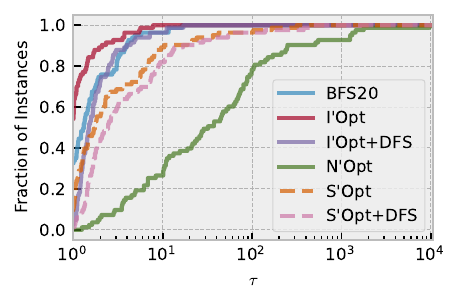}
  \caption{Time performance profile comparing our approaches with \textsc{BFS20} by Borowitz~\etal\cite{DBLP:conf/acda/BorowitzG023}. Algorithm names are abbreviated.}
  \label{fig:results}
\end{subfigure}\hfill
  \begin{subfigure}[t]{0.3\textwidth}
  \centering
  \includegraphics[width=0.95\textwidth]{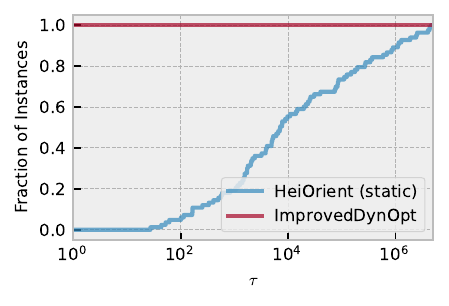}
  \caption{Time performance profile comparing the update time of our fastest approach and one-time solve using \textsc{HeiOrient (static)} on \hbox{the final graphs.}}
  \label{fig:results:fast}
\end{subfigure}
\end{figure*}
\section{Experimental Evaluation}
\label{sec:experiments}
\label{s:exp}
In the following we compare our new algorithms on a benchmark set curated by Borowitz~\etal\cite{DBLP:conf/acda/BorowitzG023}. We also include the state-of-art dynamic inexact algorithm with the highest quality from Borowitz~\etal\cite{DBLP:conf/acda/BorowitzG023} and the static exact algorithm by~Reinstädtler~\etal\cite{reinstaedtler2024engineering}.

\subsection{Optimizations for Path Finding Algorithms.}\label{sec:path_finding}
To enhance the performance of our fully dynamic algorithms, we now present optimizations for the path search. The resulting path-finding algorithms are detailed in Algorithm~\ref{alg:find:path} in the Appendix.
Building upon insights from Reinstädtler~\etal\cite{reinstaedtler2024engineering}, we limit the exploring phase of the path search. Specifically, when searching for an improving path from a node $u$ with out-degree $d$, we propose exploring only vertices with an out-degree of $d - 1$. Traversing over vertices with out-degree $d$ is unnecessary, as Invariant~\ref{strong_invariant}~and~\ref{improved_invariant} ensure that no improving path can be found through these nodes.
Reinstädtler~et~al.~\cite{reinstaedtler2024engineering} also suggest using a depth-first search with an early check, meaning all neighbors are examined first to determine if their out-degree is smaller than $d - 1$ when processing a node.
Other techniques presented in \cite{reinstaedtler2024engineering}, such as the shared visited array, are applicable only to the \textsc{TightenOutdegree} subroutine when using a depth-first search. 
With the shared visited array-technique unsuccessful explored paths are marked visited until no improvement is found. 
 This prevents re-exploring vertices.
Furthermore, we suspect that the length of the improving paths are relatively short.
Hence, we are testing breadth-first search for the path finding routines in Algorithm~\ref{alg:find:path} as an alternative to the \hbox{depth-first search approach.}

\noindent \textbf{\textit{Methodology.}}
We performed our implementations using C++ and compiled them using gcc 9.4 with full optimization turned on (-O3 flag). 
All of our experiments were run on a  machine 
equipped with one AMD EPYC 7702P 64-Core Processor running at 2.0GHz with 256MB L3 Cache and  1TB of main memory.

By default, we perform five repetitions per instance (graph) and execute up to 32 experiments in parallel. The order of experiments was random.
We measure the \emph{total time} taken to compute \emph{all edge insertions and deletions} and when comparing with static algorithms compare the \emph{average update time} (total time divided by number of updates) with the static running time.
Generally, we use the \emph{geometric mean} when averaging over different instances
in order to give every instance a comparable influence on the final result. 
In order to compare different algorithms, we use \emph{performance profiles}~\cite{DBLP:journals/mp/DolanM02}.
These plots relate the objective function size / running time  of all algorithms to the corresponding objective function size / running time produced / consumed by each algorithm.
More precisely, the $y$-axis shows $\#\{\text{objective} \leq \tau * \text{best} \} / \# \text{instances}$, where objective corresponds to
the result of an algorithm on an instance and best refers to the best result of any algorithm shown within the plot.
When we look at running time, the $y$-axis shows $\#\{t \leq \tau * \text{fastest} \} / \# \text{instances}$, where $t$ corresponds to
the time of an algorithm on an instance and fastest refers to the time of the fastest algorithm on that instance.
The parameter $\tau\geq 1$ in this equation is plotted on the $x$-axis.
For each algorithm, this yields a non-decreasing, piecewise constant function.
Thus, if we are interested in the number of instances where an algorithm is the best/fastest, \hbox{we only need to look at $\tau = 1$.}

\begin{table*}
  \centering
  \caption{Average total running time over 5 runs in seconds (lower is better) 
          for a subset of instances of the test set. Full results in Table~\ref{table:results} in the Appendix. BFS20~\cite{DBLP:conf/acda/BorowitzG023} does not solve all instances optimally. The static solver HeiOrient~\cite{reinstaedtler2024engineering} is only solving the final instance. (Fastest optimal dynamic algorithm in \textbf{bold}).}
  \label{tab:sample:instances}

  {\footnotesize
  \begin{tabular}{lS[table-format=5.2]S[table-format=5.2]S[table-format=5.2]S[table-format=5.2]S[table-format=5.2]|S[table-format=5.2]S[table-format=5.2]}
    Instance&{{{\textsc{N'Opt}}}}&{{{\textsc{S'Opt+DFS}}}}&{{{\textsc{S'Opt}}}}&{{{\textsc{I'Opt+DFS}}}}&{{{\textsc{I'Opt}}}}&{{{\textsc{BFS20}}}}&{{{\textsc{HeiOrient}}}}\\
    \midrule

    auto&66266.64&29796.92&14237.42&30229.06&\bfseries 13838.68&1771.35&0.26\\
    citeulike\_ui&4.42&0.10&0.15&0.07&\bfseries 0.07&0.32&0.07\\
    delaunay\_n20&102905.40&30.56&11.19&31.03&\bfseries 10.62&35.90&0.50\\
    dewiki\_clean&12369.42&1163.47&1174.61&209.71&\bfseries 154.39&188.98&2.41\\
    fe\_tooth&372.92&77.84&\bfseries 36.72&73.09&36.91&42.09&0.40\\
    rgg\_n\_2\_17\_s0&36.62&1.94&1.25&0.22&\bfseries 0.16&0.15&0.04\\
    wiki\_simple\_en&45.10&6.39&4.23&2.31&\bfseries 1.71&1.20&0.04\\
  \end{tabular}
  }
  \end{table*}
\begin{table}[t]
\centering
\caption{Geometric mean running time (seconds) for all updates and the relative speed over all 83 instances.}
\label{tab:geo:mean}

{\footnotesize
\begin{tabular}{lr@{\hskip 10mm}r}
  Algorithm &  Time [s] & Rel. Speed \\
          \midrule
          \textsc{ImprovedDynOpt} & 1.15 &1\\
          \textsc{BFS20} (heuristic) & 1.53 &1.32\\
          \textsc{ImprovedDynOpt+DFS} & 1.69 & 1.46\\
          \textsc{StrongDynOpt} & 2.56 & 2.22\\
          \textsc{StrongDynOpt+DFS} & 3.55 & 3.08\\
          \textsc{NaiveDynOpt} & 30.35 &26.32\\
\end{tabular}
}

\end{table}

\noindent \textbf{\textit{Instances.}}
We evaluate our algorithms on a number of large graphs, that were collected by Borowitz~\etal\cite{DBLP:conf/acda/BorowitzG023} in their recent study. 
These graphs are  from various backgrounds~\cite{benchmarksfornetworksanalysis,DBLP:journals/corr/abs-2003-00736,UFsparsematrixcollection,snap,DBLP:conf/www/Kunegis13,konect:unlink,DBLP:journals/jpdc/FunkeLMPSSSL19,kappa}.
The full version of their paper \cite{borowitz2023engineeringfullydynamicdeltaorientation} summarizes the main properties of the~benchmark~set.
     
The graphs are undirected, without self-loops or parallel edges.
Static graphs are converted to be dynamic by starting with an empty graph and inserting all edges in a random order.
In the data set are also real dynamic instances -- most of these instances, however, only feature insertions (with the exception being \texttt{amazon-ratings}, \texttt{movielens10m}, \texttt{dewiki} and \texttt{wiki\_simple\_en}) as there is currently a lack of publicly available instances that also feature deletions. 
Figure~\ref{fig:instances} in the Appendix shows the temporal course of the optimal maximum out-degree  for these instances.
\paragraph{Implementations.} We use the implementation of Borowitz~\etal\cite{DBLP:conf/acda/BorowitzG023} for the \emph{inexact/heuristic} strategy \textsc{BFS20} for which they reported the best quality.
In order to verify the optimality and to compare the running times we are using the static algorithm \textsc{HeiOrient (static)} by Reinstädtler~et~al.~\cite{reinstaedtler2024engineering} on the final static graph after all update operations are executed.
We implemented the naive (\textsc{NaiveDynOpt}), strong (\textsc{StrongDynOpt}), and improved (\textsc{ImprovedDynOpt}) invariant-based algorithms for exact edge orientation using a breadth-first path approach. Additionally, we tested, whether replacing the breadth-first search with an eager depth-first search is feasible, these variants \hbox{are marked with a \textsc{+DFS}.}

\subsection{Quality Comparison.}
We validated that all our solutions are exact by running the algorithm by Reinstädtler~\etal\cite{reinstaedtler2024engineering} on the final graph of the edit sequence.
Figure~\ref{fig:results:size} shows the quality of the \textsc{BFS20} approach in comparison to our exact solvers and the \textsc{HeiOrient} by Reinstädtler~\etal\cite{reinstaedtler2024engineering} on the resulting static graph.
Around~92\% of the instances are solved optimally by \textsc{BFS20}. The remaining instances (\texttt{4elt}, \texttt{delaunay*}, \texttt{fe\_sphere}, \texttt{fe\_pwt}, \texttt{fe\_tooth}) are solved not optimally by \textsc{BFS20}, resulting in  an out-degree up to 33\% worse than the optimal solution.

\subsection{Running Time Comparison.}
\label{exp:overall}
In Figure~\ref{fig:results} we present the performance profile for the running times of our algorithms in comparison to the best inexact competitor \textsc{BFS20} by Borowitz~\etal\cite{DBLP:conf/acda/BorowitzG023}.
The \textsc{ImprovedDynOpt} approach is for 60\% of the instances the fastest approach and dominates the other algorithms in the performance profile.
Table~\ref{tab:geo:mean} reports the geometric mean running times over all instances.
Detailed running times for a representative subset can be found in Table~\ref{tab:sample:instances}, while all results are in the Appendix in Table~\ref{table:results}.
The second fastest approach is the 32\% slower inexact \textsc{BFS20} approach by Borowitz~\etal\cite{DBLP:conf/acda/BorowitzG023}.
 Interestingly, the \textsc{ImprovedDynOpt+DFS} approach has a similar profile to the \textsc{BFS20} approach. %
 The stronger invariant version is again slower, while the naive \textsc{NaiveDynOpt} is naturally the slowest algorithm overall.
 On one instance from online networks (\texttt{t60k}) it needs more than three order magnitude longer than the fastest approach. 
  On the \texttt{delaunay20} instance it requires 5 orders of magnitude more time than the~\textsc{ImprovedDynOpt}~approach.
  However, on the finite element \texttt{bcsstk29} and \texttt{bcsstk30} instances, it is the fastest algorithm by a margin of nearly a quarter. On the other \texttt{bcsstk*} instances it is considerably slower, up to a factor of 64 for the \texttt{bcsstk32} instance. 

 On some instances from finite element background (\texttt{fe\_body}, \texttt{fe\_tooth}, \texttt{wing\_nodal}) the \textsc{StrongDynOpt} approach is faster than all other approaches.
 These instances have an average degree close to the maximum out-degree.
 Citation networks like \texttt{citeulike\_ui} and \texttt{citation*} are solved by \textsc{ImprovedDynOpt} up to four times faster than \textsc{BFS20}, while the related, denser \texttt{coAuthors*} and \texttt{coPapers*} instances are solved the fastest by \textsc{BFS20}.

 In general, the methods employing a depth-first search are slower than the breadth-first approaches.
  Only on two instances (\texttt{add*}) with very low running times across all approaches they are the fastest. 
On two of four instances with deletions (\texttt{simple\_wiki\_en} and \texttt{movielens10m}) the \textsc{BFS20} is faster, while on the \texttt{dewiki\_clean} and \texttt{amazon-ratings} instance the improved \textsc{ImprovedDynOpt} variant is faster. 

\subsection{Comparison with Static Exact Algorithms.}
In Figure~\ref{fig:results:fast} we compare the average update time of our and Borowitz'~\cite{DBLP:conf/acda/BorowitzG023} approaches with a one time solve using the resulting static graph and \textsc{HeiOrient} by Reinstädtler~et~al.~\cite{reinstaedtler2024engineering}.
 The total running times of the dynamic algorithms were normalized by dividing by the number of operations.
 The \textsc{HeiOrient} approach is up to 6 orders of magnitude slower than our approaches. 
 In the geometric mean over all instances HeiOrient is  \numprint{14307.80} times slower in this metric.

\vspace{-0.25cm}
\section{Conclusion}
In this paper we have introduced two  invariant based  algorithms and a naive algorithm for solving the fully dynamic exact $\Delta$-orientation problem in general.
By relaxing the invariants we can provide a faster algorithm.
In experiments, we have shown that our best algorithm \textsc{ImprovedDynOpt} is 32\% faster than the previous \emph{inexact}  state-of-the-art while always maintaining the optimal solution.
The update time per operation is for our algorithm up to 6 orders of magnitude faster than the state-of-the-art static solver.
Future avenues of work include  parallelization, batching of modifications and adjustments to the objective function such that for example the sum of the squared out-degrees is~to~be~optimized.

\section*{Acknowledgements} 
We acknowledge support by DFG grant \hbox{SCHU 2567/3-1}. Moreover, we like to acknowledge Dagstuhl Seminar 22461 on dynamic graph algorithms.

\nprounddigits{0}
{
        \footnotesize
\bibliographystyle{plainnat}
}
\bibliography{paperfixed}

\begin{thebibliography}{43}
\providecommand{\natexlab}[1]{#1}
\providecommand{\url}[1]{\texttt{#1}}
\expandafter\ifx\csname urlstyle\endcsname\relax
  \providecommand{\doi}[1]{doi: #1}\else
  \providecommand{\doi}{doi: \begingroup \urlstyle{rm}\Url}\fi

\bibitem[Aichholzer et~al.(1995)Aichholzer, Aurenhammer, and
  Rote]{aichholzer1995optimal}
Oswin Aichholzer, Franz Aurenhammer, and G{\"u}nter Rote.
\newblock \emph{{O}ptimal {G}raph {O}rientation {W}ith {S}torage
  {A}pplications}.
\newblock Universit{\"a}t Graz/Technische Universit{\"a}t Graz. SFB
  F003-Optimierung und Kontrolle, 1995.

\bibitem[Asahiro et~al.(2007)Asahiro, Miyano, Ono, and
  Zenmyo]{asahiro2007graph}
Yuichi Asahiro, Eiji Miyano, Hirotaka Ono, and Kouhei Zenmyo.
\newblock {G}raph {O}rientation {A}lgorithms to minimize the {M}aximum
  {O}utdegree.
\newblock \emph{Int. J. Found. Comput. Sci.}, 18\penalty0 (2):\penalty0
  197--215, 2007.
\newblock \doi{10.1142/S0129054107004644}.
\newblock URL \url{https://doi.org/10.1142/S0129054107004644}.

\bibitem[Bader et~al.(2014)Bader, Meyerhenke, Sanders, Schulz, Kappes, and
  Wagner]{benchmarksfornetworksanalysis}
David~A. Bader, Henning Meyerhenke, Peter Sanders, Christian Schulz, Andrea
  Kappes, and Dorothea Wagner.
\newblock {B}enchmarking for {G}raph {C}lustering and {P}artitioning.
\newblock In \emph{Encyclopedia of Social Network Analysis and Mining}, pages
  73--82. 2014.
\newblock \doi{10.1007/978-1-4614-6170-8\_23}.
\newblock URL \url{https://doi.org/10.1007/978-1-4614-6170-8\_23}.

\bibitem[Banerjee et~al.(2019)Banerjee, Raman, and Saurabh]{BANERJEE20201}
Niranka Banerjee, Venkatesh Raman, and Saket Saurabh.
\newblock {F}ully {D}ynamic {A}rboricity {M}aintenance.
\newblock In Ding{-}Zhu Du, Zhenhua Duan, and Cong Tian, editors,
  \emph{Computing and Combinatorics - 25th International Conference, {COCOON}
  2019, Xi'an, China, July 29-31, 2019, Proceedings}, volume 11653 of
  \emph{Lecture Notes in Computer Science}, pages 1--12, 2019.
\newblock \doi{10.1007/978-3-030-26176-4\_1}.
\newblock URL \url{https://doi.org/10.1007/978-3-030-26176-4\_1}.

\bibitem[Bergamini and Meyerhenke(2015)]{DBLP:conf/esa/BergaminiM15}
Elisabetta Bergamini and Henning Meyerhenke.
\newblock {F}ully-dynamic {A}pproximation of {B}etweenness {C}entrality.
\newblock In Nikhil Bansal and Irene Finocchi, editors, \emph{Algorithms -
  {ESA} 2015 - 23rd Annual European Symposium, Patras, Greece, September 14-16,
  2015, Proceedings}, volume 9294 of \emph{Lecture Notes in Computer Science},
  pages 155--166, 2015.
\newblock \doi{10.1007/978-3-662-48350-3\_14}.
\newblock URL \url{https://doi.org/10.1007/978-3-662-48350-3\_14}.

\bibitem[Berglin and Brodal(2020)]{DBLP:journals/algorithmica/BerglinB20}
Edvin Berglin and Gerth~St{\o}lting Brodal.
\newblock {A} {S}imple {G}reedy {A}lgorithm for {D}ynamic {G}raph
  {O}rientation.
\newblock \emph{Algorithmica}, 82\penalty0 (2):\penalty0 245--259, 2020.
\newblock \doi{10.1007/S00453-018-0528-0}.
\newblock URL \url{https://doi.org/10.1007/s00453-018-0528-0}.

\bibitem[Blumenstock(2016)]{doi:10.1137/1.9781611974317.10}
Markus Blumenstock.
\newblock {F}ast {A}lgorithms for {P}seudoarboricity.
\newblock In Michael~T. Goodrich and Michael Mitzenmacher, editors,
  \emph{Proceedings of the Eighteenth Workshop on Algorithm Engineering and
  Experiments, {ALENEX} 2016, Arlington, Virginia, USA, January 10, 2016},
  pages 113--126, 2016.
\newblock \doi{10.1137/1.9781611974317.10}.
\newblock URL \url{https://doi.org/10.1137/1.9781611974317.10}.

\bibitem[Borowitz et~al.(2023{\natexlab{a}})Borowitz, Gro{\ss}mann, and
  Schulz]{DBLP:conf/acda/BorowitzG023}
Jannick Borowitz, Ernestine Gro{\ss}mann, and Christian Schulz.
\newblock {E}ngineering {F}ully {D}ynamic {$\Delta$}-orientation {A}lgorithms.
\newblock In \emph{{ACDA}}, pages 25--37, 2023{\natexlab{a}}.

\bibitem[Borowitz et~al.(2023{\natexlab{b}})Borowitz, Großmann, and
  Schulz]{borowitz2023engineeringfullydynamicdeltaorientation}
Jannick Borowitz, Ernestine Großmann, and Christian Schulz.
\newblock {E}ngineering {F}ully {D}ynamic $\delta$-orientation {A}lgorithms,
  2023{\natexlab{b}}.
\newblock URL \url{https://arxiv.org/abs/2301.06968}.

\bibitem[Brodal and Fagerberg(1999)]{DBLP:conf/wads/BrodalF99}
Gerth~St{\o}lting Brodal and Rolf Fagerberg.
\newblock {D}ynamic {R}epresentation of {S}parse {G}raphs.
\newblock In Frank K. H.~A. Dehne, Arvind Gupta, J{\"{o}}rg{-}R{\"{u}}diger
  Sack, and Roberto Tamassia, editors, \emph{Algorithms and Data Structures,
  6th International Workshop, {WADS} '99, Vancouver, British Columbia, Canada,
  August 11-14, 1999, Proceedings}, volume 1663 of \emph{Lecture Notes in
  Computer Science}, pages 342--351, 1999.
\newblock \doi{10.1007/3-540-48447-7\_34}.
\newblock URL \url{https://doi.org/10.1007/3-540-48447-7\_34}.

\bibitem[Charikar(2000)]{charikar2000greedy}
Moses Charikar.
\newblock {G}reedy {A}pproximation {A}lgorithms for {F}inding {D}ense
  {C}omponents in {A} {G}raph.
\newblock In Klaus Jansen and Samir Khuller, editors, \emph{Approximation
  Algorithms for Combinatorial Optimization, Third International Workshop,
  {APPROX} 2000, Saarbr{\"{u}}cken, Germany, September 5-8, 2000, Proceedings},
  volume 1913 of \emph{Lecture Notes in Computer Science}, pages 84--95, 2000.
\newblock \doi{10.1007/3-540-44436-X\_10}.
\newblock URL \url{https://doi.org/10.1007/3-540-44436-X\_10}.

\bibitem[Chekuri et~al.(2024)Chekuri, Christiansen, Holm, {van der Hoog},
  Quanrud, Rotenberg, and Schwiegelshohn]{DBLP:journals/corr/abs-2209-14087}
Chandra Chekuri, Aleksander Bj{\o}rn~Grodt Christiansen, Jacob Holm, Ivor {van
  der Hoog}, Kent Quanrud, Eva Rotenberg, and Chris Schwiegelshohn.
\newblock {A}daptive {O}ut-orientations with {A}pplications.
\newblock In David~P. Woodruff, editor, \emph{Proceedings of the 2024
  {ACM-SIAM} Symposium on Discrete Algorithms, {SODA} 2024, Alexandria, VA,
  USA, January 7-10, 2024}, pages 3062--3088, 2024.
\newblock \doi{10.1137/1.9781611977912.110}.
\newblock URL \url{https://doi.org/10.1137/1.9781611977912.110}.

\bibitem[Chen et~al.(1994)Chen, Matsumoto, Wang, Zhang, and
  Zhang]{DBLP:journals/gc/ChenMWZZ94}
Boliong Chen, Makoto Matsumoto, Jianfang Wang, Zhongfu Zhang, and Jianxun
  Zhang.
\newblock {A} {S}hort {P}roof of {N}ash-{W}illiams' {T}heorem for the
  {A}rboricity of a {G}raph.
\newblock \emph{Graphs Comb.}, 10\penalty0 (1):\penalty0 27--28, 1994.
\newblock \doi{10.1007/BF01202467}.
\newblock URL \url{https://doi.org/10.1007/BF01202467}.

\bibitem[Christiansen and Rotenberg(2022)]{DBLP:conf/icalp/ChristiansenR22}
Aleksander B.~G. Christiansen and Eva Rotenberg.
\newblock {F}ully-dynamic {\(\alpha\)} + 2 {A}rboricity {D}ecompositions and
  {I}mplicit {C}olouring.
\newblock In Mikolaj Bojanczyk, Emanuela Merelli, and David~P. Woodruff,
  editors, \emph{49th International Colloquium on Automata, Languages, and
  Programming, {ICALP} 2022, July 4-8, 2022, Paris, France}, volume 229 of
  \emph{LIPIcs}, pages 42:1--42:20, 2022.
\newblock \doi{10.4230/LIPIcs.ICALP.2022.42}.
\newblock URL \url{https://doi.org/10.4230/LIPIcs.ICALP.2022.42}.

\bibitem[Davis()]{UFsparsematrixcollection}
Timothy Davis.
\newblock {The {U}niversity of {F}lorida {S}parse {M}atrix {C}ollection,
  \url{http://www.cise.ufl.edu/research/sparse/matrices}, 2008}.
\newblock URL \url{http://www.cise.ufl.edu/research/sparse/matrices/}.

\bibitem[Demetrescu and Italiano(2006)]{DBLP:journals/talg/DemetrescuI06}
Camil Demetrescu and Giuseppe~F. Italiano.
\newblock {E}xperimental {A}nalysis of {D}ynamic {A}ll {P}airs {S}hortest
  {P}ath {A}lgorithms.
\newblock \emph{{ACM} Trans. Algorithms}, 2\penalty0 (4):\penalty0 578--601,
  2006.
\newblock \doi{10.1145/1198513.1198519}.
\newblock URL \url{https://doi.org/10.1145/1198513.1198519}.

\bibitem[Dolan and Mor{\'{e}}(2002)]{DBLP:journals/mp/DolanM02}
Elizabeth~D. Dolan and Jorge~J. Mor{\'{e}}.
\newblock {B}enchmarking {O}ptimization {S}oftware {W}ith {P}erformance
  {P}rofiles.
\newblock \emph{Math. Program.}, 91\penalty0 (2):\penalty0 201--213, 2002.
\newblock \doi{10.1007/S101070100263}.
\newblock URL \url{https://doi.org/10.1007/s101070100263}.

\bibitem[Doll et~al.(2011)Doll, Hartmann, and Wagner]{DBLP:conf/wads/DollHW11}
Christof Doll, Tanja Hartmann, and Dorothea Wagner.
\newblock {F}ully-dynamic {H}ierarchical {G}raph {C}lustering {U}sing {C}ut
  {T}rees.
\newblock In Frank Dehne, John Iacono, and J{\"{o}}rg{-}R{\"{u}}diger Sack,
  editors, \emph{Algorithms and Data Structures - 12th International Symposium,
  {WADS} 2011, New York, NY, USA, August 15-17, 2011. Proceedings}, volume 6844
  of \emph{Lecture Notes in Computer Science}, pages 338--349, 2011.
\newblock \doi{10.1007/978-3-642-22300-6\_29}.
\newblock URL \url{https://doi.org/10.1007/978-3-642-22300-6\_29}.

\bibitem[Frigioni et~al.(1998)Frigioni, Ioffreda, Nanni, and
  Pasqualone]{DBLP:journals/jea/FrigioniINP98}
Daniele Frigioni, Mario Ioffreda, Umberto Nanni, and Giulio Pasqualone.
\newblock {E}xperimental {A}nalysis of {D}ynamic {A}lgorithms for the
  {S}ingle-source {S}hortest-path {P}roblem.
\newblock \emph{{ACM} J. Exp. Algorithmics}, 3:\penalty0 5, 1998.
\newblock \doi{10.1145/297096.297147}.
\newblock URL \url{https://doi.org/10.1145/297096.297147}.

\bibitem[Funke et~al.(2019)Funke, Lamm, Meyer, Penschuck, Sanders, Schulz,
  Strash, and von Looz]{DBLP:journals/jpdc/FunkeLMPSSSL19}
Daniel Funke, Sebastian Lamm, Ulrich Meyer, Manuel Penschuck, Peter Sanders,
  Christian Schulz, Darren Strash, and Moritz von Looz.
\newblock {C}ommunication-free {M}assively {D}istributed {G}raph {G}eneration.
\newblock \emph{J. Parallel Distributed Comput.}, 131:\penalty0 200--217, 2019.
\newblock \doi{10.1016/J.JPDC.2019.03.011}.
\newblock URL \url{https://doi.org/10.1016/j.jpdc.2019.03.011}.

\bibitem[Goranci et~al.(2021)Goranci, Henzinger, Leniowski, Schulz, and
  Svozil]{DBLP:conf/alenex/GoranciHLSS21}
Gramoz Goranci, Monika Henzinger, Dariusz Leniowski, Christian Schulz, and
  Alexander Svozil.
\newblock {F}ully {D}ynamic \emph{k}-center {C}lustering in {L}ow {D}imensional
  {M}etrics.
\newblock In Martin Farach{-}Colton and Sabine Storandt, editors,
  \emph{Proceedings of the Symposium on Algorithm Engineering and Experiments,
  {ALENEX} 2021, Virtual Conference, January 10-11, 2021}, pages 143--153,
  2021.
\newblock \doi{10.1137/1.9781611976472.11}.
\newblock URL \url{https://doi.org/10.1137/1.9781611976472.11}.

\bibitem[Hanauer et~al.(2020{\natexlab{a}})Hanauer, Henzinger, and
  Schulz]{DBLP:conf/alenex/HanauerH020}
Kathrin Hanauer, Monika Henzinger, and Christian Schulz.
\newblock {F}ully {D}ynamic {S}ingle-source {R}eachability in {P}ractice: {A}n
  {E}xperimental {S}tudy.
\newblock In Guy~E. Blelloch and Irene Finocchi, editors, \emph{Proceedings of
  the Symposium on Algorithm Engineering and Experiments, {ALENEX} 2020, Salt
  Lake City, UT, USA, January 5-6, 2020}, pages 106--119, 2020{\natexlab{a}}.
\newblock \doi{10.1137/1.9781611976007.9}.
\newblock URL \url{https://doi.org/10.1137/1.9781611976007.9}.

\bibitem[Hanauer et~al.(2020{\natexlab{b}})Hanauer, Henzinger, and
  Schulz]{DBLP:conf/wea/HanauerH020}
Kathrin Hanauer, Monika Henzinger, and Christian Schulz.
\newblock {F}aster {F}ully {D}ynamic {T}ransitive {C}losure in {P}ractice.
\newblock In Simone Faro and Domenico Cantone, editors, \emph{18th
  International Symposium on Experimental Algorithms, {SEA} 2020, June 16-18,
  2020, Catania, Italy}, volume 160 of \emph{LIPIcs}, pages 14:1--14:14,
  2020{\natexlab{b}}.
\newblock \doi{10.4230/LIPICS.SEA.2020.14}.
\newblock URL \url{https://doi.org/10.4230/LIPIcs.SEA.2020.14}.

\bibitem[Hanauer et~al.(2021)Hanauer, Henzinger, and
  Schulz]{DBLP:journals/corr/abs-2102-11169}
Kathrin Hanauer, Monika Henzinger, and Christian Schulz.
\newblock {R}ecent {A}dvances in {F}ully {D}ynamic {G}raph {A}lgorithms.
\newblock \emph{CoRR}, abs/2102.11169, 2021.
\newblock URL \url{https://arxiv.org/abs/2102.11169}.

\bibitem[Hanauer et~al.(2022)Hanauer, Henzinger, and
  Hua]{DBLP:conf/sand/HanauerHH22}
Kathrin Hanauer, Monika Henzinger, and Qi~Cheng Hua.
\newblock {F}ully {D}ynamic {F}our-vertex {S}ubgraph {C}ounting.
\newblock In James Aspnes and Othon Michail, editors, \emph{1st Symposium on
  Algorithmic Foundations of Dynamic Networks, {SAND} 2022, March 28-30, 2022,
  Virtual Conference}, volume 221 of \emph{LIPIcs}, pages 18:1--18:17, 2022.
\newblock \doi{10.4230/LIPICS.SAND.2022.18}.
\newblock URL \url{https://doi.org/10.4230/LIPIcs.SAND.2022.18}.

\bibitem[He et~al.(2014)He, Tang, and Zeh]{DBLP:conf/isaac/HeTZ14}
Meng He, Ganggui Tang, and Norbert Zeh.
\newblock {O}rienting {D}ynamic {G}raphs, with {A}pplications to {M}aximal
  {M}atchings and {A}djacency {Q}ueries.
\newblock In Hee{-}Kap Ahn and Chan{-}Su Shin, editors, \emph{Algorithms and
  Computation - 25th International Symposium, {ISAAC} 2014, Jeonju, Korea,
  December 15-17, 2014, Proceedings}, volume 8889 of \emph{Lecture Notes in
  Computer Science}, pages 128--140, 2014.
\newblock \doi{10.1007/978-3-319-13075-0\_11}.
\newblock URL \url{https://doi.org/10.1007/978-3-319-13075-0\_11}.

\bibitem[Henzinger et~al.(2022)Henzinger, Noe, and
  Schulz]{DBLP:conf/alenex/HenzingerN022}
Monika Henzinger, Alexander Noe, and Christian Schulz.
\newblock {P}ractical {F}ully {D}ynamic {M}inimum {C}ut {A}lgorithms.
\newblock In Cynthia~A. Phillips and Bettina Speckmann, editors,
  \emph{Proceedings of the Symposium on Algorithm Engineering and Experiments,
  {ALENEX} 2022, Alexandria, VA, USA, January 9-10, 2022}, pages 13--26, 2022.
\newblock \doi{10.1137/1.9781611977042.2}.
\newblock URL \url{https://doi.org/10.1137/1.9781611977042.2}.

\bibitem[Holtgrewe et~al.(2010)Holtgrewe, Sanders, and Schulz]{kappa}
Manuel Holtgrewe, Peter Sanders, and Christian Schulz.
\newblock {E}ngineering a {S}calable {H}igh {Q}uality {G}raph {P}artitioner.
\newblock In \emph{24th {IEEE} International Symposium on Parallel and
  Distributed Processing, {IPDPS} 2010, Atlanta, Georgia, USA, 19-23 April 2010
  - Conference Proceedings}, pages 1--12, 2010.
\newblock \doi{10.1109/IPDPS.2010.5470485}.
\newblock URL \url{https://doi.org/10.1109/IPDPS.2010.5470485}.

\bibitem[Kannan et~al.(1992)Kannan, Naor, and
  Rudich]{DBLP:journals/siamdm/KannanNR92}
Sampath Kannan, Moni Naor, and Steven Rudich.
\newblock {I}mplicit {R}epresentation of {G}raphs.
\newblock \emph{{SIAM} J. Discret. Math.}, 5\penalty0 (4):\penalty0 596--603,
  1992.
\newblock \doi{10.1137/0405049}.
\newblock URL \url{https://doi.org/10.1137/0405049}.

\bibitem[Kopelowitz et~al.(2014)Kopelowitz, Krauthgamer, Porat, and
  Solomon]{DBLP:conf/icalp/KopelowitzKPS14}
Tsvi Kopelowitz, Robert Krauthgamer, Ely Porat, and Shay Solomon.
\newblock {O}rienting {F}ully {D}ynamic {G}raphs with {W}orst-case {T}ime
  {B}ounds.
\newblock In Javier Esparza, Pierre Fraigniaud, Thore Husfeldt, and Elias
  Koutsoupias, editors, \emph{Automata, Languages, and Programming - 41st
  International Colloquium, {ICALP} 2014, Copenhagen, Denmark, July 8-11, 2014,
  Proceedings, Part {II}}, volume 8573 of \emph{Lecture Notes in Computer
  Science}, pages 532--543, 2014.
\newblock \doi{10.1007/978-3-662-43951-7\_45}.
\newblock URL \url{https://doi.org/10.1007/978-3-662-43951-7\_45}.

\bibitem[Kowalik(2006)]{kowalik2006approximation}
Lukasz Kowalik.
\newblock {A}pproximation {S}cheme for {L}owest {O}utdegree {O}rientation and
  {G}raph {D}ensity {M}easures.
\newblock In Tetsuo Asano, editor, \emph{Algorithms and Computation, 17th
  International Symposium, {ISAAC} 2006, Kolkata, India, December 18-20, 2006,
  Proceedings}, volume 4288 of \emph{Lecture Notes in Computer Science}, pages
  557--566, 2006.
\newblock \doi{10.1007/11940128\_56}.
\newblock URL \url{https://doi.org/10.1007/11940128\_56}.

\bibitem[Kowalik(2007)]{DBLP:journals/ipl/Kowalik07}
Lukasz Kowalik.
\newblock {A}djacency {Q}ueries in {D}ynamic {S}parse {G}raphs.
\newblock \emph{Inf. Process. Lett.}, 102\penalty0 (5):\penalty0 191--195,
  2007.
\newblock \doi{10.1016/J.IPL.2006.12.006}.
\newblock URL \url{https://doi.org/10.1016/j.ipl.2006.12.006}.

\bibitem[Kowalik and Kurowski(2006)]{DBLP:journals/talg/KowalikK06}
Lukasz Kowalik and Maciej Kurowski.
\newblock {O}racles for {B}ounded-length {S}hortest {P}aths in {P}lanar
  {G}raphs.
\newblock \emph{{ACM} Trans. Algorithms}, 2\penalty0 (3):\penalty0 335--363,
  2006.
\newblock \doi{10.1145/1159892.1159895}.
\newblock URL \url{https://doi.org/10.1145/1159892.1159895}.

\bibitem[Krommidas and Zaroliagis(2008)]{DBLP:journals/jea/KrommidasZ08}
Ioannis Krommidas and Christos~D. Zaroliagis.
\newblock {A}n {E}xperimental {S}tudy of {A}lgorithms for {F}ully {D}ynamic
  {T}ransitive {C}losure.
\newblock \emph{{ACM} J. Exp. Algorithmics}, 12:\penalty0 1.6:1--1.6:22, 2008.
\newblock \doi{10.1145/1227161.1370597}.
\newblock URL \url{https://doi.org/10.1145/1227161.1370597}.

\bibitem[Kunegis(2013)]{DBLP:conf/www/Kunegis13}
J{\'{e}}r{\^{o}}me Kunegis.
\newblock {KONECT:} the {K}oblenz {N}etwork {C}ollection.
\newblock In Leslie Carr, Alberto H.~F. Laender, Bernadette~Farias
  L{\'{o}}scio, Irwin King, Marcus Fontoura, Denny Vrandecic, Lora Aroyo,
  Jos{\'{e}} Palazzo~M. de~Oliveira, Fernanda Lima, and Erik Wilde, editors,
  \emph{22nd International World Wide Web Conference, {WWW} '13, Rio de
  Janeiro, Brazil, May 13-17, 2013, Companion Volume}, pages 1343--1350, 2013.
\newblock \doi{10.1145/2487788.2488173}.
\newblock URL \url{https://doi.org/10.1145/2487788.2488173}.

\bibitem[Leskovec()]{snap}
Jure Leskovec.
\newblock {S}tanford {N}etwork {A}nalysis {P}ackage ({S}{N}{A}{P}).
\newblock \url{http://snap.stanford.edu/index.html}.

\bibitem[Nash-Williams(1961)]{10.1112/jlms/s1-36.1.445}
C.~St.J.~A. Nash-Williams.
\newblock {Edge-Disjoint {S}panning {T}rees of {F}inite {G}raphs}.
\newblock \emph{Journal of the London Mathematical Society}, s1-36\penalty0
  (1):\penalty0 445--450, 01 1961.
\newblock ISSN 0024-6107.
\newblock \doi{10.1112/jlms/s1-36.1.445}.
\newblock URL \url{https://doi.org/10.1112/jlms/s1-36.1.445}.

\bibitem[Nash-Williams(1964)]{10.1112/jlms/s1-39.1.12}
C.~St.J.~A. Nash-Williams.
\newblock {Decomposition of {F}inite {G}raphs {I}nto {F}orests}.
\newblock \emph{Journal of the London Mathematical Society}, s1-39\penalty0
  (1):\penalty0 12--12, 01 1964.
\newblock ISSN 0024-6107.
\newblock \doi{10.1112/jlms/s1-39.1.12}.
\newblock URL \url{https://doi.org/10.1112/jlms/s1-39.1.12}.

\bibitem[Neiman and Solomon(2016)]{DBLP:journals/talg/NeimanS16}
Ofer Neiman and Shay Solomon.
\newblock {S}imple {D}eterministic {A}lgorithms for {F}ully {D}ynamic {M}aximal
  {M}atching.
\newblock \emph{{ACM} Trans. Algorithms}, 12\penalty0 (1):\penalty0 7:1--7:15,
  2016.
\newblock \doi{10.1145/2700206}.
\newblock URL \url{https://doi.org/10.1145/2700206}.

\bibitem[Penschuck et~al.(2020)Penschuck, Brandes, Hamann, Lamm, Meyer, Safro,
  Sanders, and Schulz]{DBLP:journals/corr/abs-2003-00736}
Manuel Penschuck, Ulrik Brandes, Michael Hamann, Sebastian Lamm, Ulrich Meyer,
  Ilya Safro, Peter Sanders, and Christian Schulz.
\newblock {R}ecent {A}dvances in {S}calable {N}etwork {G}eneration.
\newblock \emph{CoRR}, abs/2003.00736, 2020.
\newblock URL \url{https://arxiv.org/abs/2003.00736}.

\bibitem[Preusse et~al.(2013)Preusse, Kunegis, Thimm, Staab, and
  Gottron]{konect:unlink}
Julia Preusse, J{\'{e}}r{\^{o}}me Kunegis, Matthias Thimm, Steffen Staab, and
  Thomas Gottron.
\newblock {S}tructural {D}ynamics of {K}nowledge {N}etworks.
\newblock In Emre Kiciman, Nicole~B. Ellison, Bernie Hogan, Paul Resnick, and
  Ian Soboroff, editors, \emph{Proceedings of the Seventh International
  Conference on Weblogs and Social Media, {ICWSM} 2013, Cambridge,
  Massachusetts, USA, July 8-11, 2013}, 2013.
\newblock URL
  \url{http://www.aaai.org/ocs/index.php/ICWSM/ICWSM13/paper/view/6076}.

\bibitem[Reinstädtler et~al.(2024)Reinstädtler, Schulz, and
  Uçar]{reinstaedtler2024engineering}
H.~Reinstädtler, C.~Schulz, and B.~Uçar.
\newblock {E}ngineering {E}dge {O}rientation {A}lgorithms, 2024.

\bibitem[Venkateswaran(2004)]{venkateswaran2004minimizing}
Venkat Venkateswaran.
\newblock {M}inimizing {M}aximum {I}ndegree.
\newblock \emph{Discret. Appl. Math.}, 143\penalty0 (1-3):\penalty0 374--378,
  2004.
\newblock \doi{10.1016/J.DAM.2003.07.007}.
\newblock URL \url{https://doi.org/10.1016/j.dam.2003.07.007}.

\end{thebibliography}
\clearpage
\begin{appendix}
  
  \section{Instances}
  \begin{figure}[H]
    \centering
    \includegraphics[width=0.33\textwidth]{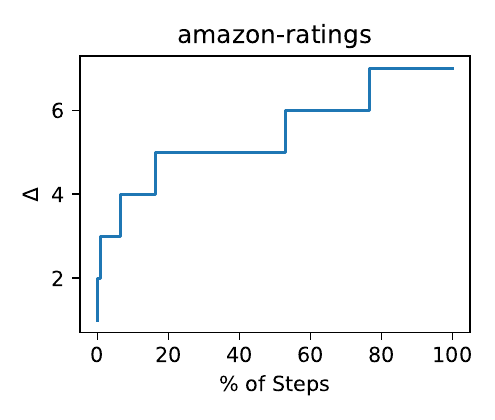}
    \includegraphics[width=0.33\textwidth]{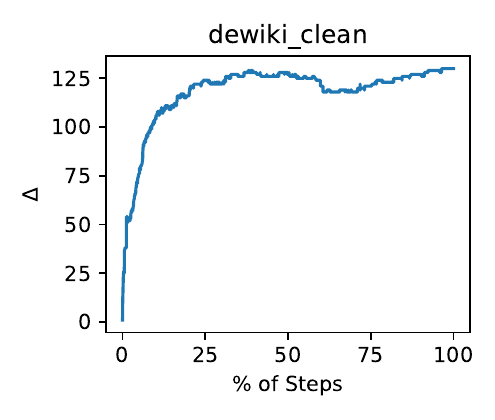}
    \includegraphics[width=0.33\textwidth]{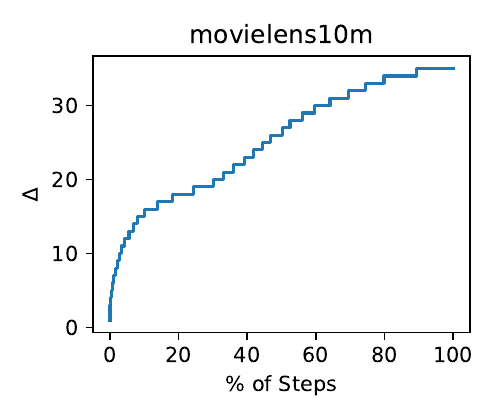}
    \includegraphics[width=0.33\textwidth]{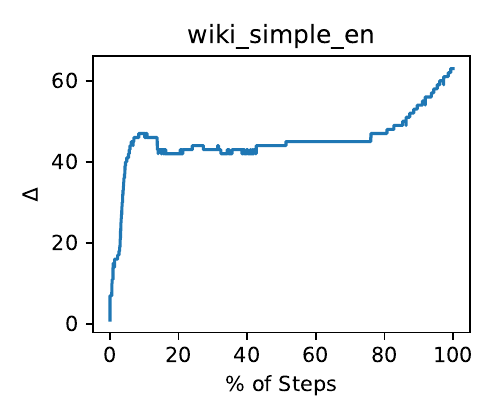}
    \caption{Optimal maximum out-degree ($\Delta$) over time for the fully dynamic instances.}
    \label{fig:instances}
  \end{figure}
  \clearpage
  \onecolumn
\section{Detailed Results}
{
\begin{longtblr}[
caption = {Average total running time over five runs in seconds (lower is better). BFS20~\cite{DBLP:conf/acda/BorowitzG023} is a heuristic algorithm and does not solve all instances optimally. The static solver HeiOrient~\cite{reinstaedtler2024engineering} is only solving the final instance. Algorithm names are abbreviated. The fastest optimal dynamic algorithm is highlighted in \textbf{bold} front.},
  label = {table:results},
 ]{colsep = 1mm,
rowsep = 1mm,fontsize=\tiny,
colspec = {lS[table-format=5.2]S[table-format=5.2]S[table-format=5.2]S[table-format=5.2]S[table-format=5.2]||S[table-format=5.2]S[table-format=5.2]},rowhead =1,  row{even} = {gray9}, row{1} = {font=\scshape}, column{1}={font=\ttfamily}}
&{{{N'Opt}}}&{{{S'Opt+DFS}}}&{{{S'Opt}}}&{{{I'Opt+DFS}}}&{{{I'Opt}}}&{{{BFS20}}}&{{{HeiOrient}}}\\
144&4752.03&4982.90&2551.56&5021.59&\bfseries 2540.59&451.76&0.06\\
3elt&0.31&0.01&\bfseries 0.01&0.01&0.01&0.04&0.00\\
4elt&5.48&0.06&0.04&0.05&\bfseries 0.04&0.15&0.00\\
598a&1493.04&590.08&195.06&529.99&\bfseries 186.96&133.99&0.27\\
PGPgiantcompo&0.05&0.02&0.01&0.01&\bfseries 0.01&0.01&0.00\\
RHG-1m-nodes-10m-e'&4459.00&30.00&22.77&4.32&\bfseries 3.55&10.87&0.90\\
RHG-1m-nodes-20m-e'&16873.14&175.50&129.92&33.01&\bfseries 13.53&222.34&1.66\\
add20&0.01&0.00&0.00&\bfseries 0.00&0.01&0.01&0.00\\
add32&0.04&0.00&0.00&\bfseries 0.00&0.00&0.00&0.00\\
amazon-2008&5843.61&8869.86&3899.14&168.11&\bfseries 23.50&34.90&0.57\\
amazon-ratings&2.73&0.19&0.29&0.14&\bfseries 0.13&0.18&0.16\\
as-22july06&0.09&0.01&0.01&0.01&\bfseries 0.01&0.03&0.00\\
as-skitter&153.71&45.04&34.82&3.99&\bfseries 3.35&15.28&0.40\\
auto&66266.64&29796.92&14237.42&30229.06&\bfseries 13838.68&1771.35&0.26\\
bcsstk29&\bfseries 64.01&102.96&68.94&103.00&71.38&34.11&0.04\\
bcsstk30&331.52&235.76&152.03&178.38&\bfseries 108.87&74.01&1.60\\
bcsstk31&15.19&11.40&8.10&0.51&\bfseries 0.34&0.48&0.04\\
bcsstk32&14.07&103.26&68.73&0.28&\bfseries 0.22&0.35&0.04\\
bcsstk33&\bfseries 46.30&91.08&60.19&90.97&62.33&37.06&0.33\\
brack2&219.35&30.62&17.12&33.44&\bfseries 16.60&26.71&0.35\\
citationCiteseer&177.40&11.83&5.06&4.58&\bfseries 1.67&5.43&0.18\\
citeulike\_ui&4.42&0.10&0.15&0.07&\bfseries 0.07&0.32&0.07\\
cnr-2000&41.91&20.55&15.43&0.92&\bfseries 0.81&1.42&0.11\\
coAuthorsCiteseer&4.24&1.15&0.88&\bfseries 0.17&0.17&0.16&0.06\\
coAuthorsDBLP&5.63&1.40&1.10&0.28&\bfseries 0.28&0.20&0.08\\
coPapersCiteseer&2551.82&971.90&728.58&163.82&\bfseries 121.18&34.11&0.58\\
coPapersDBLP&662.30&516.80&393.50&11.18&\bfseries 9.38&9.32&0.68\\
crack&1.32&0.02&0.02&0.02&\bfseries 0.02&0.07&0.00\\
cs4&6.33&0.02&0.02&0.01&\bfseries 0.01&0.13&0.00\\
cti&6.62&0.98&0.69&0.97&\bfseries 0.65&0.87&0.00\\
data&0.17&0.06&\bfseries 0.05&0.06&0.05&0.07&0.00\\
delaunay\_n16&97.03&0.35&0.22&0.34&\bfseries 0.19&0.79&0.01\\
delaunay\_n17&832.98&1.01&0.76&1.00&\bfseries 0.58&2.27&0.04\\
delaunay\_n20&102905.40&30.56&11.19&31.03&\bfseries 10.62&35.90&0.50\\
dewiki-2013&2533.29&167.08&122.98&25.02&\bfseries 20.55&118.14&2.51\\
dewiki\_clean&12369.42&1163.47&1174.61&209.71&\bfseries 154.39&188.98&2.41\\
dnc-temporalGraph&0.01&0.00&0.00&0.00&\bfseries 0.00&0.00&0.00\\
email-EuAll&0.50&0.29&\bfseries 0.22&0.29&0.22&0.17&0.00\\
enron&5.56&2.04&1.47&1.88&\bfseries 1.40&0.81&0.02\\
eu-2005&656.55&1469.95&1086.25&12.91&\bfseries 10.54&13.49&0.75\\
facebook-wosn-wall&3.55&0.34&0.25&0.17&\bfseries 0.11&0.15&0.01\\
fe\_4elt2&2.50&0.03&\bfseries 0.03&0.03&0.03&0.09&0.00\\
fe\_body&49.60&0.94&\bfseries 0.53&0.94&0.54&0.98&0.01\\
fe\_ocean&1131.13&34.75&23.56&34.88&\bfseries 17.83&13.12&0.01\\
fe\_pwt&45.47&4.32&\bfseries 2.70&4.55&2.90&1.04&0.01\\
fe\_rotor&1163.31&957.36&\bfseries 547.99&887.74&556.19&239.51&0.04\\
fe\_sphere&6.51&0.09&0.07&0.10&\bfseries 0.07&0.19&0.00\\
fe\_tooth&372.92&77.84&\bfseries 36.72&73.09&36.91&42.09&0.40\\
finan512&108.12&4.15&2.08&4.06&\bfseries 1.99&3.02&0.03\\
flickr-growth&22917.82&2746.43&1811.54&1488.92&\bfseries 945.98&838.96&4.60\\
haggle&0.00&0.00&0.00&0.00&\bfseries 0.00&0.00&0.00\\
in-2004&4459.32&1108.66&756.32&468.09&\bfseries 350.28&94.23&0.61\\
lastfm\_band&44.64&16.88&\bfseries 11.57&16.93&12.03&12.72&0.08\\
ljournal-2008&6975.45&2207.74&1491.73&102.76&\bfseries 79.32&40.87&6.91\\
lkml-reply&2.25&0.47&0.33&0.47&\bfseries 0.33&0.36&0.01\\
loc-brightkite\_edges&1.33&0.33&0.26&0.13&\bfseries 0.11&0.12&0.01\\
loc-gowalla\_edges&19.95&2.22&1.70&1.20&\bfseries 0.93&1.36&0.08\\
m14b&13633.64&7607.11&4081.14&7961.83&\bfseries 3949.02&721.81&0.27\\
memplus&0.20&0.07&0.06&0.07&\bfseries 0.06&0.06&0.00\\
movielens10m&24.94&2.66&1.89&2.99&\bfseries 1.80&1.09&0.03\\
munmun\_digg&1.67&0.17&0.12&0.18&\bfseries 0.11&0.17&0.01\\
p2p-Gnutella04&0.82&0.63&0.41&0.63&\bfseries 0.40&0.60&0.00\\
proper\_loans&1228.51&210.34&137.46&313.92&\bfseries 121.28&112.49&0.19\\
rgg\_n\_2\_15\_s0&3.96&0.15&0.12&0.03&\bfseries 0.03&0.03&0.01\\
rgg\_n\_2\_16\_s0&11.43&0.48&0.35&0.06&\bfseries 0.05&0.06&0.01\\
rgg\_n\_2\_17\_s0&36.62&1.94&1.25&0.22&\bfseries 0.16&0.15&0.04\\
soc-Slashdot0902&5.10&6.99&5.26&1.07&\bfseries 0.75&1.34&0.01\\
sociopatterns-infections&0.01&0.00&0.00&0.00&\bfseries 0.00&0.00&0.00\\
s'exchange-s'overflow&54.44&1.02&0.93&1.02&\bfseries 0.62&0.89&0.16\\
t60k&48.77&0.03&0.03&0.03&\bfseries 0.03&0.06&0.00\\
topology&0.34&0.10&0.08&0.09&\bfseries 0.07&0.13&0.01\\
uk&0.15&\bfseries 0.00&0.00&0.00&0.00&0.00&0.00\\
vibrobox&3.69&1.16&0.78&0.68&\bfseries 0.51&0.44&0.01\\
wave&4272.21&3082.34&1556.88&2999.15&\bfseries 1555.95&416.36&0.08\\
web-Google&52.78&5.33&4.17&0.60&\bfseries 0.60&0.73&0.26\\
whitaker3&1.85&0.04&\bfseries 0.03&0.03&0.03&0.08&0.00\\
wiki-Talk&107.46&35.03&25.26&34.39&\bfseries 25.26&17.42&0.11\\
wiki\_simple\_en&45.10&6.39&4.23&2.31&\bfseries 1.71&1.20&0.04\\
wikipedia-growth&6019.03&79.22&60.24&36.54&\bfseries 24.55&66.99&3.10\\
wing&52.53&0.07&0.05&0.07&\bfseries 0.04&0.47&0.01\\
wing\_nodal&7.55&6.38&\bfseries 3.83&6.41&3.90&5.01&0.07\\
wordassociation-2011&2.42&2.48&1.81&2.48&\bfseries 1.81&1.49&0.00\\
youtube-u-growth&963.10&20.40&15.73&17.77&\bfseries 9.29&28.87&1.86\\
\end{longtblr}

}
\end{appendix}

\end{document}